\documentclass[reqno]{amsart}
\usepackage{hyperref}

\usepackage{amsmath,amssymb} 
\pdfoutput = 1 
\usepackage{amsmath}

\usepackage{euscript}
\usepackage[T2A]{fontenc}
\usepackage[utf8]{inputenc}

\usepackage{textcomp}
\usepackage{mathtools}

\usepackage{amsthm}
\theoremstyle{plain}
\newtheorem{theorem}{Theorem}[section]

\theoremstyle{definition}

\newtheorem{primer}{Example}[section]

\theoremstyle{remark}

\numberwithin{equation}{section}

\allowdisplaybreaks

\RequirePackage{caption}
\usepackage{graphicx}
\graphicspath{}
\DeclareGraphicsExtensions{.pdf,.png,.jpg}

\begin{document}
	
	\title[ Exact solution Ising-like $2 \times 2 \times \infty$ models with multispin interactions 
	]
	{ Exact solution Ising-like $2 \times 2 \times \infty$ models with multispin interactions  
	}
	\author[P. Khrapov]{Pavel Khrapov}
	\address{Pavel Khrapov \\ Department of Mathematics
		\\ Bauman Moscow State Technical University \\  ul. Baumanskaya 2-ya, 5/1, Moscow \\ 105005, Moscow,  Russian Federation}  
	\email{khrapov@bmstu.ru , pvkhrapov@gmail.com }

	\author[V. Nikishina]{Veronika Nikishina}
    \address{Veronika Nikishina \\ Department of Mathematics
    	\\ Bauman Moscow State Technical University \\  ul. Baumanskaya 2-ya, 5/1, Moscow \\ 105005, Moscow,  Russian Federation}  
    \email{nikishinava@student.bmstu.ru , veronika\_nikishina98@mail.ru}

	
	\begin{abstract}
		 
		The paper considers the generalized Ising model in the $2\times2\times\infty$ strip with a Hamiltonian invariant with respect to the central axis of rotation through the angle $\pi/2$, which includes all possible multiplicative interactions of an even number of spins in the unit cube.The exact value of free energy and heat capacity in the thermodynamic limit is found.  For percolation invariant with respect to rotation about the central axis of rotation through the angle $\pi/2$, a limit relation for the non-percolation probability is derived. The solution was obtained by the transfer matrix method. In the general case, finding the largest eigenvalue of a $16 \times 16$ transfer matrix reduces to solving a quartic equation using the Ferrari method. In special cases, switching the transfer matrix of models with auxiliary matrices made it possible to reduce the problem to quadratic and linear equations. 
		 A separate chapter is devoted to the gonigendrial model in the $2\times2\times\infty$ strip, the exact values of the free energy and heat capacity in the thermodynamic limit are found for the case of free boundary conditions and an analogue of cyclically closed boundary conditions in both directions perpendicular to the central axis of rotation. The numerical solution for the $3\times3\times\infty$ gonigendrial model in the case of free boundary conditions and in the case of cyclic closure is compared with the corresponding analytical solutions in the volume $2\times2\times\infty$. The closeness of free energy and heat capacity in the thermodynamic limit for the gonigendrial model on the entire three-dimensional lattice and for the analog of cyclic closure in the $2\times2\times\infty$ strip is shown. The exact value of the free energy and heat capacity in the thermodynamic limit in the $2\times2\times\infty$ strip is obtained for the Antiferromagnetic layered Ising model, the similarity of these physical characteristics with the numerical solution in an infinite volume is shown.
		 
	\end{abstract}
	
	\maketitle
	
	\tableofcontents

	\section{Introduction}\label{in}
	
	The Ising model was introduced back in the 20s of the 20-th entury \cite{Ising} to describe magnetic transitions, but is still the subject of intense research.  In particular, the Ising model with interactions of nearest neighbors ($NN$), next nearest neighbors ($NNN$), and plaquette interactions ($P$) on
	cubic lattice was considered as a simple model of the statistical mechanics of random surfaces \cite{Cappi,Karowski}, microemulsions \cite{Seto}, where critical fluctuations in the density of water droplets were studied, as well as discretized string action (the so-called gonigendrial model) \cite{Ambartzumian}-\cite{Nishiyama}.

	The Ising model on a $2\times2\times\infty$ lattice with a simple cubic cell and under the condition of equal interactions in both transverse directions can be reduced to an $n\times\infty$ Ising model consisting of $n$ cells closed into a cylinder, if we put the number of cells equal to $n=4$. The exact value of the partition function and free energy for the nearest neighbor Ising model without an external magnetic field on the $n\times\infty$ cylinder was first obtained by Onsager \cite{Onsager}, and the result for lattice $2\times2\times\infty$
	generalized to the case of different interaction coefficients in all three directions in \cite{Yurishchev,Yurishchev_1997}.In work \cite{Ratner} all 16 eigenvalues for the 3D Ising model are found under the condition that the interactions in both transverse directions are equal. The interaction of next nearest neighbors is taken into account in \cite{Krijanovskii}, where the eigenvalues of the transfer matrix of the Ising model are found.
	
	 The exact value of free energy and heat capacity for the Ising model with NN, NNN interactions was obtained in \cite{Yokota_Terufumi_1988} and \cite{Yokota_Terufumi_1989}, and phase diagrams were plotted for various ratios of parameters. 
	   The author of \cite{Yurishchev_2007} presents exact analytical formulas for free energy for two $2\times2\times\infty$ Ising models. The first model \cite{Yurishchev_2007} has a simple cubic lattice with fully anisotropic interactions. The second model \cite{Yurishchev_2007} consists of two different types of linear chains and includes non-intersecting diagonal bonds on the side faces of the $2\times2\times\infty$ parallelepiped. In both cases, the solutions are expressed in terms of square radicals and are obtained using the obvious symmetry of the Hamiltonians and the hidden algebraic symmetry of the characteristic polynomial of the transfer matrix.  
	The two-dimensional Ising model has been intensively studied, the Hamiltonian of which additionally takes into account the interaction of the next nearest neighbors \cite{Binder},\cite{Kalz},\cite{dos Anjos}. A three-dimensional analog of this model was studied by authors of \cite{Ramazanov} using the Monte Carlo method.

 	In \cite{Coniglio} percolation in the Ising model is studied, based on the generalized results of \cite{Miyamoto}. Based on percolation in a finite-width strip for an independent field and the Ising model, results are obtained in \cite{Khrapov_P_V}, for continuous models in
	\cite{Minlos_R_A_Khrapov_P_V}. A meaningful review and results on percolation theory can be found in \cite{Men'shikov_M_V_Molchanov_S_A_Sidorenko_A_F}.
	In recent years, interest in simple systems of Ising spins has arisen as a result of considering the random surface model in the context of string theory \cite{Savvidy,Johnston}. This model was called gonigendrial and its three-dimensional case was studied by approximation methods 	
	\cite{Pelizzola}, Monte Carlo method \cite{Dimopoulos}, using differential operators \cite{dosAnjos} and using Novotny’s transfer-matrix method \cite{Nishiyama}. 
		
	In this paper, we consider the generalized Ising model on the lattice $2\times2\times\infty$ with a Hamiltonian invariant with respect to the central axis of rotation on $\pi/2$, which includes all possible interactions of an even number of spins: double, quadruple, six-fold and eight-dimensional interactions . For a model with the most general Hamiltonian of this type, the exact value of the free energy in the thermodynamic limit is written out. Let us note that finding the free energy in the thermodynamic limit is reduced to finding the largest eigenvalue of the transfer matrix, which is the root of the fourth degree equation. However, in the same way, we can find all the eigenvalues of the transfer matrix, this makes it possible to obtain the partition function and free energy in a finite closed chain $2\times2\times L$.
	The authors introduced a special type of percolation, which is invariant under rotation through the angle $\pi/2$, and describes the limiting behavior of the non-percolation probability. Next, we consider two special cases in which the search for the largest eigenvalue of the transfer matrix is reduced to solving a quadratic equation. An exact solution is obtained for the gonigendrial model with free boundary conditions and in the cyclically closed case. An exact solution of the Antiferromagnetic layered Ising model is found.
	This work is a logical continuation of the works 
	\cite{Minlos_R_A_Khrapov_P_V}, \cite{Khrapov_P_V},
	 \cite{Khrapov1}, \cite{Khrapov2},  \cite{Khrapov_Golosov_Nikishina} according to Ising-like models.
	
	There is a following structure in the article.
	The chapter \ref{Model_description} introduces the basic concepts: model Hamiltonian, partition function, thermodynamic characteristics and invariant percolation. In chapter \ref{exact_sol}, using the method of transfer matrices and taking into account the commutation of the resulting transfer matrix with the "rotation" operator $\;$ on $\pi/2$, expressions for the thermodynamic characteristics of the model are written. In chapter \ref{sec_perc} the non-percolation probability is written out, which is obtained taking into account the invariance of the introduced type of percolation with respect to the rotation by $\pi/2$. The chapter \ref{Simplest cases} is devoted to the consideration of some of the simplest cases of solution, which are obtained taking into account the commutation of a particular form of the transfer matrix with some other matrices. In the \ref{gonihedric} chapter, an analytical solution is written for the generalized gonigendrial model in the case of free boundary conditions and in the case of an analogue of cyclic closure, a numerical solution is found in a certain range of parameters in the volume $3\times3\times\infty$ for free boundary conditions and in the case of cyclic closure. The numerical solution for the $3\times3\times\infty$ gonigendrial model in the case of free boundary conditions and in the case of cyclic closure is compared with the corresponding analytical solutions. The closeness of free energy and heat capacity in the thermodynamic limit is shown for the gonigendrial model on the entire three-dimensional lattice and for the analogue of cyclic closure on the lattice $2\times2\times\infty$. In chapter \ref{NN_NNN_interactions} the exact value of free energy and heat capacity in the thermodynamic limit for the Antiferromagnetic layered Ising model is found.

		\section{Model description} 
	\label{Model_description}
	
	\subsection{Model Hamiltonian and Physical Characteristics}  \label{Ham}
	
	Let us consider a three-dimensional square cyclically closed lattice model with size $2\times 2\times M$ (Fig.1),
	\begin{equation}\label{L22M}
		\mathcal{L}_{2,2,M}=\{t_i^m, m=0,1,...,M-1,t_i^M \equiv t_i^0, i=0,1,2,3\},
	\end{equation}
	the total number of the lattice sites is $N=4M$.

	\begin{center}
		\begin{picture}(250,210) 
			\put(0,50){\line(1,0){250}}
			\put(0,86){\line(1,0){250}}
			\put(0,150){\line(1,0){250}}
			\put(0,186){\line(1,0){250}}
			\put(52,50){\line(0,1){100}}
			\put(100,86){\line(0,1){100}}
			\put(152,50){\line(0,1){100}}
			\put(200,86){\line(0,1){100}}
			\put(52,50){\line(4,3){48}}
			\put(152,50){\line(4,3){48}}
			\put(52,150){\line(4,3){48}}
			\put(152,150){\line(4,3){48}}
			\put(48,47){\line(-4,-3){20}}
			\put(104,89){\line(4,3){20}}
			\put(148,47){\line(-4,-3){20}}
			\put(204,89){\line(4,3){20}}
			\put(48,147){\line(-4,-3){20}}
			\put(104,189){\line(4,3){20}}
			\put(148,147){\line(-4,-3){20}}
			\put(204,189){\line(4,3){20}}
			\put(52,50){\circle*{4}}
			\put(100,86){\circle*{4}}
			\put(54,40){$t^m_0$}
			\put(102,76){$t^m_1$}
			\put(152,50){\circle*{4}}
			\put(200,86){\circle*{4}}
			\put(154,40){$t^{m+1}_0$}
			\put(202,76){$t^{m+1}_1$}
			\put(52,150){\circle*{4}}
			\put(100,186){\circle*{4}}
			\put(54,140){$t^m_3$}
			\put(102,176){$t^m_2$}
			\put(152,150){\circle*{4}}
			\put(200,186){\circle*{4}}
			\put(154,140){$t^{m+1}_3$}
			\put(202,176){$t^{m+1}_2$}
			\label{fig1}
			\put(75,60){$x$}
			\put(45,95){$y$}
			\put(110,40){$z$}
		\end{picture}
	\end{center}
	\begin{center}
		FIGURE 1. Arrangement of spins in a strip and with a cross section $2\times2$.
		The product of transfer matrices $\tau$ is taken from left to right
	\end{center}
	
	Let's assume that a particle is localed at each site $t_i^m$ . The state of the particle is determined by the value of the
	spin $\sigma^{m}_i\equiv\sigma_{t_i^m}, i=0,1,2,3, m=0,1,...,M-1$. 
	
	Let
	\begin{equation}
		\Omega=\{t_{i_1}^{m_1},t_{i_2}^{m_2},...,t_{i_k}^{m_k}\}\subset\mathcal{L}_{2,2,M}.
	\end{equation}

	Let's introduce 
	\begin{equation}
		\sigma_\Omega=\prod_{t\in\Omega}\sigma_t.
	\end{equation}
	
	Let $\Omega^m_C=\{t_{i}^{m},t_{i}^{m+1},i=0,1,2,3\}$ be a cube.
	We introduce the set of elementary generating supports of the Hamiltonian
		\begin{equation}\label{Phi^m}
		\begin{gathered}
			\Phi_2^m=\{  \{t_{0}^{m},t_{1}^{m} \},  \{t_{0}^{m},t_{2}^{m} \},  \{t_{0}^{m+1},t_{1}^{m+1} \},  \{t_{0}^{m+1},t_{2}^{m+1} \},\\ \{t_{0}^{m},t_{0}^{m+1} \}, \{t_{0}^{m},t_{1}^{m+1} \}, \{t_{0}^{m},t_{2}^{m+1} \}, \{t_{0}^{m},t_{3}^{m+1} \}\}, \\
				\Phi_4^m=\{ \{t_{0}^{m},t_{0}^{m+1} , t_{1}^{m+1},t_{2}^{m+1} \}, 
			  \{t_{0}^{m},t_{1}^{m+1} ,t_{2}^{m+1},t_{3}^{m+1} \},\\
			   \{t_{0}^{m},t_{2}^{m+1} , t_{3}^{m+1},t_{0}^{m+1} \},
			    \{t_{0}^{m},t_{3}^{m+1} , t_{0}^{m+1},t_{1}^{m+1} \},\\
 \{t_{0}^{m},t_{1}^{m} , t_{0}^{m+1},t_{1}^{m+1} \}, 
 \{t_{0}^{m},t_{1}^{m} , t_{1}^{m+1},t_{2}^{m+1} \},\\ 
  \{t_{0}^{m},t_{1}^{m} , t_{2}^{m+1},t_{3}^{m+1} \}, 
   \{t_{0}^{m},t_{1}^{m} , t_{3}^{m+1},t_{0}^{m+1} \}, \\
    \{t_{0}^{m},t_{1}^{m} , t_{0}^{m+1},t_{2}^{m+1} \}, 
\{t_{0}^{m},t_{1}^{m} , t_{1}^{m+1},t_{3}^{m+1} \}, \\
\{t_{0}^{m},t_{2}^{m} , t_{0}^{m+1},t_{1}^{m+1} \}, 
\{t_{0}^{m},t_{2}^{m} , t_{1}^{m+1},t_{2}^{m+1} \},\\ 
\{t_{0}^{m},t_{2}^{m} , t_{2}^{m+1},t_{3}^{m+1} \}, 
\{t_{0}^{m},t_{2}^{m} , t_{3}^{m+1},t_{0}^{m+1} \}, \\
\{t_{0}^{m},t_{2}^{m} , t_{0}^{m+1},t_{2}^{m+1} \}, 
\{t_{0}^{m},t_{2}^{m} , t_{1}^{m+1},t_{3}^{m+1} \}, \\
\{t_{0}^{m},t_{1}^{m} , t_{2}^{m},t_{0}^{m+1} \}, 
\{t_{0}^{m},t_{1}^{m} , t_{2}^{m},t_{1}^{m+1} \},\\ 
\{t_{0}^{m},t_{1}^{m} , t_{2}^{m},t_{2}^{m+1} \}, 
\{t_{0}^{m},t_{1}^{m} , t_{2}^{m},t_{3}^{m+1} \}, \\
\{t_{0}^{m},t_{1}^{m} , t_{2}^{m},t_{3}^{m} \}, 
\{t_{0}^{m+1},t_{1}^{m+1} , t_{2}^{m+1},t_{3}^{m+1} \}, \\
\Phi_6^m=\{\{ \Omega^m_C\setminus \{t_{0}^{m},t_{1}^{m} \}\},  \{ \Omega^m_C\setminus \{t_{0}^{m},t_{2}^{m} \}\},  \\
\{ \Omega^m_C\setminus \{t_{0}^{m+1},t_{1}^{m+1} \}\},  \{ \Omega^m_C\setminus \{t_{0}^{m+1},t_{2}^{m+1} \}\},\\ \{ \Omega^m_C\setminus \{t_{0}^{m},t_{0}^{m+1} \}\}, \{ \Omega^m_C\setminus \{t_{0}^{m},t_{1}^{m+1} \}\}, \{ \Omega^m_C\setminus \{t_{0}^{m},t_{2}^{m+1} \}\},\\
\{ \Omega^m_C\setminus \{t_{0}^{m},t_{3}^{m+1} \}\}\}, \\
\Phi_8^m=\{\{t_{0}^{m},t_{1}^{m} , t_{2}^{m},t_{3}^{m},
t_{0}^{m+1},t_{1}^{m+1} , t_{2}^{m+1},t_{3}^{m+1} \}\}, \\
\Phi^m=\Phi_2^m \cup \Phi_4^m \cup \Phi_6^m \cup \Phi_8^m  .
		\end{gathered}	
	\end{equation}
	
	Let us define the elementary generating Hamiltonian
		\begin{equation}
		\begin{gathered} 
			\mathcal{\hat{H}}^m=\mathcal{\hat{H}}^m_2+\mathcal{\hat{H}}^m_4+\mathcal{\hat{H}}^m_6+\mathcal{\hat{H}}^m_8,	
		\end{gathered} 	
	\end{equation}
	where 
	\begin{equation}
		\begin{gathered} 
			\mathcal{\hat{H}}_r^m=-\sum_{{\alpha}^m  \in \Phi_r^m }{J_{{\alpha}^m }} \sigma_{{\alpha}^m },	
			r=2,4,6,8,
		\end{gathered} 	
	\end{equation}
	components of the Hamiltonian responsible for the interaction of an even number of $r$ spins in the cube of $\Omega^m_C$.
	
	For function $f(\sigma^{m}_0,\sigma^{m}_1,\sigma^{m}_2,\sigma^{m}_3,\sigma^{m+1}_0,\sigma^{m+1}_1,\sigma^{m+1}_2,\sigma^{m+1}_3)$, depending on $\sigma_t, t\in\Omega^m_C$, we introduce the operator of rotation $R_{\pi/2}^m$  by an angle ${\pi/2}$
	\begin{equation}\label{R}
		\begin{gathered}
		R^m_{\pi/2}(f(\sigma^{m}_0,\sigma^{m}_1,\sigma^{m}_2,\sigma^{m}_3,\sigma^{m+1}_0,\sigma^{m+1}_1,\sigma^{m+1}_2,\sigma^{m+1}_3))=\\
		f(\sigma^{m}_1,\sigma^{m}_2,\sigma^{m}_3,\sigma^{m}_0,\sigma^{m+1}_1,\sigma^{m+1}_2,\sigma^{m+1}_3,\sigma^{m+1}_0).
\end{gathered}	
\end{equation}
	
	The Hamiltonian model has the form: 
	\begin{equation}\label{hamiltonian}
	\begin{gathered} 
		\mathcal H=\sum_{m=0}^{M-1}\mathcal{H}^m,			
	\end{gathered} 	
	\end{equation}
where 
\begin{equation}\label{H_m}
	\begin{gathered} 
	 \mathcal{H}^m=\mathcal{\hat{H}}^m+R^m_{\pi/2}\mathcal{\hat{H}}^m+(R^m_{\pi/2})^2\mathcal{\hat{H}}^m+(R^m_{\pi/2})^3\mathcal{\hat{H}}^m.		
	\end{gathered} 	
\end{equation}

	The partition function can be written in the following form:
	
	\begin{equation}\label{part_func}
		\begin{gathered}
			Z_{2,2,M}=\sum_{\sigma}\exp(-\mathcal H/(k_BT)), 
		\end{gathered}
	\end{equation}
	summation perfomed over all spins.  
	
	Let us determine the physical characteristics of the system in the thermodynamic limit in accordance with\cite{baxter2016}.  The free energy of the system per one lattice site is determined in the standard way	
	
	\begin{equation}\label{free_energy}
			f(T)=-k_BT\lim_{N\to \infty} N^{-1}\ln{Z_{2,2,M}(T)},
	\end{equation}
	where $ N=4M$ is the number of sites of the considered lattice.
	
	The internal energy per one lattice site is equal to 
	
	\begin{equation}\label{int_energy}
		u(T)=-T^2\frac{\partial}{\partial T}[f(T)/T].
	\end{equation}
	
	The heat capacity per one lattice site is defined as
	\begin{equation}\label{heat_capacity}
		C(T)=\frac{\partial}{\partial T}u(T)=-(2T\frac{\partial}{\partial T}[f(T)/T]+T^2\frac{\partial^2}{\partial T^2}[f(T)/T]).
	\end{equation}
	
	\subsection{Invariant percolation}  \label{Perc}
	Spin configurations in the cyclically closed strip $\mathcal{L}_{2,2,M}$ (\ref{L22M}) in which there is a column or row of $-1$: $\sigma_0^m=\sigma_{1}^m=-1$, $\sigma_0^m=\sigma_{3}^m=-1$,  $\sigma_1^m=\sigma_{2}^m=-1$, $\sigma_2^m=\sigma_{3}^m=-1$ for some $m \in \{0,...,M-1\}$, will be called percolation configurations.
	Accordingly, the non-percolation probability is the probability that there will be no percolation in the model.
	The probability of non-percolation of the model is determined by the ratio of two quantities:
	\begin{equation}\label{no_percolation}
		P_M=\frac{Z'_{2,2,M}}{Z_{2,2,M}},
	\end{equation}
	where ${Z_{2,2,M}}$ is the partition function  (\ref{part_func}) of the model with the followind Hamiltonian  (\ref{hamiltonian}), ${Z'_{2,2,M}}$ - analogue of the partition function, in which the summation is performed only over non-percolation configurations.

	\section{Exact solution for free energy in the strip $2 \times 2 \times \infty$ }\label{exact_sol}

	With considering (\ref{hamiltonian}),(\ref{H_m}) the partition function (\ref{part_func}) can be written as:

	\begin{equation}\label{new_part_func}
			Z_{2,2,M}=
			\sum_{\sigma}\exp(-\sum_{m=0}^{M-1}\mathcal{H}^m/(k_BT) ). 
	\end{equation}

	Let us introduce a transfer matrix $\theta$ with size $16\times16$ with matrix elements (Fig.2 )
	   \begin{equation}\label{teta_k,l}
	   	\begin{gathered}
	   	\theta_{k,l}= 	\theta_{(\sigma_0^m,\sigma_1^m,\sigma_2^m,\sigma_3^m),(\sigma_0^{m+1},\sigma_1^{m+1},\sigma_2^{m+1},\sigma_3^{m+1})}=\exp(-\sum_{m=0}^{M-1}\mathcal{H}^m/(k_BT) ) ,\\ 
	   		k=(1-\sigma_0^m)/2+(1-\sigma_1^m)+
	   		2 (1-\sigma_2^m) + 4 (1-\sigma_3^m), \\ 
	   		l=(1-\sigma_0^{m+1})/2+(1-\sigma_1^{m+1})+
	   		2 (1-\sigma_2^{m+1}) + 4 (1-\sigma_3^{m+1}).\\	   	
	   	\end{gathered}
	   \end{equation}
	   
	      Then
	      
	      	\begin{equation}\label{new_part_func2}
	      	Z_{2,2,M}= tr (\theta^M). 
	      \end{equation}
	      
	   To find the free energy (\ref{free_energy}) in the thermodynamic limit, it suffices to find the largest eigenvalue $\lambda_{\max}$ of the $\theta$ transfer matrix. Further, it will be shown that this largest eigenvalue $\lambda_{\max}$ is the largest root of the characteristic equation of the matrix $\tau$ with size $4\times4$:
	       \begin{equation}\label{tau}
	       	\begin{gathered}		
	       		\tau =\begin {psmallmatrix} 
	       		\begin{tabular}{p{1.5cm} p{3cm} p{1.5cm}  p{1.5cm}} \tiny{
	       				\begin{center}$\theta_{1,1}+\theta_{1,16}$ \end{center}}& \tiny{\begin{center}$\theta_{1,2}+\theta_{1,3}+\theta_{1,5}+\theta_{1,8}+ \theta_{1,9}+\theta_{1,12}+\theta_{1,14}+\theta_{1,15}$ \end{center}}& \tiny{\begin{center}$\theta_{1,4}+\theta_{1,7}+\theta_{1,10}+\theta_{1,13}$ \end{center}}&\tiny{\begin{center} $\theta_{1,6}+\theta_{1,11}$\end{center}}\\
	       			\tiny{\begin{center}$\theta_{2,1}+\theta_{2,16}$ \end{center}}&\tiny{ \begin{center}$\theta_{2,2}+\theta_{2,3}+\theta_{2,5}+\theta_{2,8}+ \theta_{2,9}+\theta_{2,12}+\theta_{2,14}+\theta_{2,15}$ \end{center}}&\tiny{ \begin{center}$\theta_{2,4}+\theta_{2,7}+\theta_{2,10}+\theta_{2,13}$ \end{center}}&\tiny{\begin{center} $\theta_{2,6}+\theta_{2,11}$\end{center}}\\
	       			\tiny{\begin{center}$\theta_{4,1}+\theta_{4,16}$ \end{center}}&\tiny{ \begin{center}$\theta_{4,2}+\theta_{4,3}+\theta_{4,5}+\theta_{4,8}+ \theta_{4,9}+\theta_{4,12}+\theta_{4,14}+\theta_{4,15}$ \end{center}}&\tiny{ \begin{center}$\theta_{4,4}+\theta_{4,7}+\theta_{4,10}+\theta_{4,13}$ \end{center}}&\tiny{\begin{center} $\theta_{4,6}+\theta_{4,11}$\end{center}}\\
	       			\tiny{\begin{center}$\theta_{6,1}+\theta_{6,16}$ \end{center}}& \tiny{\begin{center}$\theta_{6,2}+\theta_{6,3}+\theta_{6,5}+\theta_{6,8}+ \theta_{6,9}+\theta_{6,12}+\theta_{6,14}+\theta_{6,15}$ \end{center}}& \tiny{\begin{center}$\theta_{6,4}+\theta_{6,7}+\theta_{6,10}+\theta_{6,13}$ \end{center}}&\tiny{\begin{center} $\theta_{6,6}+\theta_{6,11}$\end{center}}
	       		\end{tabular}
	       	\end{psmallmatrix},\\	
	       \end{gathered}
       \end{equation}

	\begin{center}
	\begin{picture}(292,290) 
		\put(20,30){\line(0,1){252}}
		\put(20,282){\line(1,0){265}}
		\put(285,282){\line(0,-1){252}}
		\put(285,30){\line(-1,0){265}}
		\put(32,30){\line(0,1){204}}
		\put(44,30){\line(0,1){204}}
		\put(56,30){\line(0,1){204}}
		\put(68,30){\line(0,1){252}}
		\put(80,42){\line(0,1){204}}
		\put(92,42){\line(0,1){216}}
		\put(104,42){\line(0,1){204}}
		\put(116,42){\line(0,1){228}}
		\put(128,42){\line(0,1){204}}
		\put(140,42){\line(0,1){216}}
		\put(152,42){\line(0,1){204}}
		\put(164,42){\line(0,1){240}}
		\put(176,42){\line(0,1){204}}
		\put(188,42){\line(0,1){216}}
		\put(200,42){\line(0,1){204}}
		\put(212,42){\line(0,1){228}}
		\put(224,42){\line(0,1){204}}
		\put(236,42){\line(0,1){216}}
		\put(248,42){\line(0,1){204}}
		\put(260,42){\line(0,1){240}}
		\put(20,42){\line(1,0){265}}
		\put(260,54){\line(-1,0){204}}
		\put(260,66){\line(-1,0){216}}
		\put(260,78){\line(-1,0){204}}
		\put(260,90){\line(-1,0){228}}
		\put(260,102){\line(-1,0){204}}
		\put(260,114){\line(-1,0){216}}
		\put(260,126){\line(-1,0){204}}
		\put(260,138){\line(-1,0){240}}
		\put(260,150){\line(-1,0){204}}
		\put(260,162){\line(-1,0){216}}
		\put(260,174){\line(-1,0){204}}
		\put(260,186){\line(-1,0){228}}
		\put(260,198){\line(-1,0){204}}
		\put(260,210){\line(-1,0){216}}
		\put(260,222){\line(-1,0){204}}
		\put(285,234){\line(-1,0){265}}
		\put(285,246){\line(-1,0){217}}
		\put(285,258){\line(-1,0){217}}
		\put(285,270){\line(-1,0){217}}
		\put(58,34){\tiny{$\sigma_0^m$}}
		\put(46,34){\tiny{$\sigma_1^m$}}
		\put(34,34){\tiny{$\sigma_2^m$}}
		\put(22,34){\tiny{$\sigma_3^m$}}
		\put(58,46){\small{$-$}}
		\put(58,58){\small{$+$}}
		\put(58,70){\small{$-$}}
		\put(58,82){\small{$+$}}
		\put(58,94){\small{$-$}}
		\put(58,106){\small{$+$}}
		\put(58,118){\small{$-$}}
		\put(58,130){\small{$+$}}
		\put(58,142){\small{$-$}}
		\put(58,154){\small{$+$}}
		\put(58,166){\small{$-$}}
		\put(58,178){\small{$+$}}
		\put(58,190){\small{$-$}}
		\put(58,202){\small{$+$}}
		\put(58,214){\small{$-$}}
		\put(58,226){\small{$+$}}
		\put(46,52){\small{$-$}}
		\put(46,76){\small{$+$}}
		\put(46,100){\small{$-$}}
		\put(46,124){\small{$+$}}
		\put(46,148){\small{$-$}}
		\put(46,172){\small{$+$}}
		\put(46,196){\small{$-$}}
		\put(46,220){\small{$+$}}
		\put(34,64){\small{$-$}}
		\put(34,112){\small{$+$}}
		\put(34,160){\small{$-$}}
		\put(34,208){\small{$+$}}
		\put(22,88){\small{$-$}}
		\put(22,184){\small{$+$}}
		\put(262,238){\tiny{$\sigma_0^{m+1}$}}
		\put(262,250){\tiny{$\sigma_1^{m+1}$}}
		\put(262,262){\tiny{$\sigma_2^{m+1}$}}
		\put(262,274){\tiny{$\sigma_3^{m+1}$}}
		\put(250,238){\small{$-$}}
		\put(238,238){\small{$+$}}
		\put(226,238){\small{$-$}}
		\put(214,238){\small{$+$}}
		\put(202,238){\small{$-$}}
		\put(190,238){\small{$+$}}
		\put(178,238){\small{$-$}}
		\put(166,238){\small{$+$}}
		\put(154,238){\small{$-$}}
		\put(142,238){\small{$+$}}
		\put(130,238){\small{$-$}}
		\put(118,238){\small{$+$}}
		\put(106,238){\small{$-$}}
		\put(94,238){\small{$+$}}
		\put(82,238){\small{$-$}}
		\put(70,238){\small{$+$}}
		\put(244,250){\small{$-$}}
		\put(220,250){\small{$+$}}
		\put(196,250){\small{$-$}}
		\put(172,250){\small{$+$}}
		\put(148,250){\small{$-$}}
		\put(124,250){\small{$+$}}
		\put(100,250){\small{$-$}}
		\put(76,250){\small{$+$}}
		\put(232,262){\small{$-$}}
		\put(184,262){\small{$+$}}
		\put(136,262){\small{$-$}}
		\put(88,262){\small{$+$}}
		\put(208,274){\small{$-$}}
		\put(112,274){\small{$+$}}
		\label{fig4}
	\end{picture}
\\ FIGURE 2. The structure of the transfer matrix for computing the partition function in a strip with a cross section $2\times2$
\end{center}

	We can now formulate the following main theorem.
	
\begin{theorem}\label{T1} 
	\textbf{Main theorem}
	
	Let  $\lambda_{\max} =\lambda_{\max}(T)$ be the largest eigenvalue of the characteristic equation of the fourth degree of the matrix $\tau$  (\ref{tau}) for the model with the Hamiltonian (\ref{hamiltonian}-\ref{H_m}).

	Then in the thermodynamic limit:

the free energy (\ref{free_energy}) of considering model has the form
	\begin{equation} \label{free_energy_lambda}
		f(T)=-k_BT\ln{(\lambda_{\max}(T))}/4,
	\end{equation}
the internal energy (\ref{int_energy}) per one lattice site can be represented as
\begin{equation}\label{int_energy_lambda}
	u(T)=k_BT^2\frac{\partial}{\partial T}[\ln{(\lambda_{\max}(T))}/4],
\end{equation}
the heat capacity (\ref{heat_capacity}) per one lattice site is equal to
\begin{equation}\label{heat_capacity_lambda}
	C(T)=\frac{\partial}{\partial T}u(T)=2k_BT\frac{\partial}{\partial T}[\ln{(\lambda_{\max}(T))/4}]+k_BT^2\frac{\partial^2}{\partial T^2}[\ln{(\lambda_{\max}(T))}/4],
\end{equation}
	where $k_B$ is Boltzmann’s constant.

\end{theorem}

\begin{proof}

The proof is based on commutation of the matrix $\theta$ with a matrix of a special form; this allows us to refine the form of the eigenvector corresponding to the largest eigenvalue.

The transfer matrix $\theta$ commutes with the "rotation" matrix $ \;$ on $ \pi/2$ $ D:\{\sigma_0^m\rightarrow\sigma_1^{m+1};\sigma_1^m\rightarrow\sigma_2^{m+1};\sigma_2^m\rightarrow\sigma_3^{m+1};\sigma_3^m\rightarrow\sigma_0^{m+1}, m=0,..,M-1\}$ .

\begin{equation}\label{D}
 D =\begin {psmallmatrix} 
1& 0& 0& 0& 0& 0& 0& 0& 0& 0& 0& 0& 0& 0& 0& 0 \\
0& 0& 1& 0& 0& 0& 0& 0& 0& 0& 0& 0& 0& 0& 0& 0 \\
0& 0& 0& 0& 1& 0& 0& 0& 0& 0& 0& 0& 0& 0& 0& 0 \\
0& 0& 0& 0& 0& 0& 1& 0& 0& 0& 0& 0& 0& 0& 0& 0 \\
0& 0& 0& 0& 0& 0& 0& 0& 1& 0& 0& 0& 0& 0& 0& 0 \\
0& 0& 0& 0& 0& 0& 0& 0& 0& 0& 1& 0& 0& 0& 0& 0 \\
0& 0& 0& 0& 0& 0& 0& 0& 0& 0& 0& 0& 1& 0& 0& 0 \\
0& 0& 0& 0& 0& 0& 0& 0& 0& 0& 0& 0& 0& 0& 1& 0 \\
0& 1& 0& 0& 0& 0& 0& 0& 0& 0& 0& 0& 0& 0& 0& 0 \\
0& 0& 0& 1& 0& 0& 0& 0& 0& 0& 0& 0& 0& 0& 0& 0 \\
0& 0& 0& 0& 0& 1& 0& 0& 0& 0& 0& 0& 0& 0& 0& 0 \\
0& 0& 0& 0& 0& 0& 0& 1& 0& 0& 0& 0& 0& 0& 0& 0 \\
0& 0& 0& 0& 0& 0& 0& 0& 0& 1& 0& 0& 0& 0& 0& 0 \\
0& 0& 0& 0& 0& 0& 0& 0& 0& 0& 0& 1& 0& 0& 0& 0 \\
0& 0& 0& 0& 0& 0& 0& 0& 0& 0& 0& 0& 0& 1& 0& 0 \\
0& 0& 0& 0& 0& 0& 0& 0& 0& 0& 0& 0& 0& 0& 0& 1 \\
\end{psmallmatrix} 	
\end{equation}

The eigenvectors of the matrix $D$ (\ref{D}), corresponding to the positive eigenvalue $1$, have the form:
\begin{equation}
	\begin{gathered}
		\{(0, 0, 0, 0, 0, 0, 0, 0, 0, 0, 0, 0, 0, 0, 0, 1),\\
		(0, 0, 0, 0, 0, 0, 
		0, 1, 0, 0, 0, 1, 0, 1, 1, 0),\\
		(0, 0, 0, 1, 0, 0, 1, 0, 0, 1, 0, 0, 
		1, 0, 0, 0),\\
		(0, 0, 0, 0, 0, 1, 0, 0, 0, 0, 1, 0, 0, 0, 0, 0),\\
		(0, 
		1, 1, 0, 1, 0, 0, 0, 1, 0, 0, 0, 0, 0, 0, 0),\\
		(1, 0, 0, 0, 0, 0, 0, 
		0, 0, 0, 0, 0, 0, 0, 0, 0)\}.
	\end{gathered}
\end{equation}

 The linear combination of these vectors is also an eigenvector that corresponds to the eigenvalue 1. Since the commuting matrices preserve each other's own subspaces \cite{Horn} and taking into account the central symmetry of the transfer matrix $\theta$ \cite{Andrew}, we will look for the eigenvector of the matrix $\theta$ corresponding to the largest eigenvalue in the form:
\begin{displaymath}
	(x_{1} , x_{2} , x_{2}, x_{3} , x_{2} , x_{4} , x_{3} , x_{2} , x_{2}, x_{3}  , x_{4} , x_{2} , x_{3} , x_{2}, 
	x_{2} , x_{1}),
\end{displaymath}
which, by the Perron-Frobenius theorem \cite{Perron}, has positive coordinates.

The form of the eigenvector allows you to search for the largest eigenvalue not of a matrix $\theta$ with size $16\times16$, but of a matrix $\tau $ with size $4\times4$ (\ref{tau}).

The characteristic polynomial of a matrix $\tau$ (\ref{tau}) is:

\begin{equation}\label{poly4}
	\begin{gathered}
		\lambda^4+a\lambda^3+b\lambda^2+c\lambda+d=0,
	\end{gathered}
\end{equation}
where

\begin{equation}\label{a}
	\begin{gathered}
		a=-
		(\tau_{1,1 }+
		\tau_{2,2 }+
		\tau_{3,3 } +
		\tau_{4,4 }),
	\end{gathered}
\end{equation}
\begin{equation}\label{b}
	\begin{gathered}
		b=
		-\tau_{1,2 }\tau_{2,1 }+ 
		\tau_{1,1 }\tau_{2,2 } -
		\tau_{1,3 }\tau_{3,1 } -
		\tau_{2,3 }\tau_{3,2 } +
		\tau_{1,1 }\tau_{3,3 } +\\
		\tau_{2,2 }\tau_{3,3 } -
		\tau_{1,4 }\tau_{4,1 } -
		\tau_{2,4 }\tau_{4,2 } -
		\tau_{3,4 }\tau_{4,3 } +
		\tau_{1,1 }\tau_{4,4 } +
		\tau_{2,2 }\tau_{4,4 } +
		\tau_{3,3 }\tau_{4,4 },
	\end{gathered}
\end{equation}
\begin{equation}\label{c}
	\begin{gathered}
		c=\tau_{1,3 }\tau_{2,2 }\tau_{3,1 } -
		\tau_{1,2 }\tau_{2,3 }\tau_{3,1 } -
		\tau_{1,3 }\tau_{2,1 }\tau_{3,2 } +
		\tau_{1,1 }\tau_{2,3 }\tau_{3,2 } +
		\tau_{1,2 }\tau_{2,1 }\tau_{3,3 } - \\
		\tau_{1,1 }\tau_{2,2 }\tau_{3,3 } +
		\tau_{1,4 }\tau_{2,2 }\tau_{4,1 } -
		\tau_{1,2 }\tau_{2,4 }\tau_{4,1 } +
		\tau_{1,4 }\tau_{3,3 }\tau_{4,1 } - 
		\tau_{1,3 }\tau_{3,4 }\tau_{4,1 } -\\
		\tau_{1,4 }\tau_{2,1 }\tau_{4,2 } +
		\tau_{1,1 }\tau_{2,4 }\tau_{4,2 } +
		\tau_{2,4 }\tau_{3,3 }\tau_{4,2 } -
		\tau_{2,3 }\tau_{3,4 }\tau_{4,2 } -
		\tau_{1,4 }\tau_{3,1 }\tau_{4,3 } -\\
		\tau_{2,4 }\tau_{3,2 }\tau_{4,3 } +
		\tau_{1,1 }\tau_{3,4 }\tau_{4,3 } + 
		\tau_{2,2 }\tau_{3,4 }\tau_{4,3 } +
		\tau_{1,2 }\tau_{2,1 }\tau_{4,4 } -
		\tau_{1,1 }\tau_{2,2 }\tau_{4,4 } +\\
		\tau_{1,3 }\tau_{3,1 }\tau_{4,4 } +
		\tau_{2,3 }\tau_{3,2 }\tau_{4,4 } -
		\tau_{1,1 }\tau_{3,3 }\tau_{4,4 } -
		\tau_{2,2 }\tau_{3,3 }\tau_{4,4 },
	\end{gathered}
\end{equation}
\begin{equation}\label{d}
	\begin{gathered}
		d=
		\tau_{1,4 }\tau_{2,3 }\tau_{3,2 }\tau_{4,1 } -
		\tau_{1,3 }\tau_{2,4 }\tau_{3,2 }\tau_{4,1 } -
		\tau_{1,4 }\tau_{2,2 }\tau_{3,3 }\tau_{4,1} +
		\tau_{1,2 }\tau_{2,4 }\tau_{3,3 }\tau_{4,1} +\\
		\tau_{1,3 }\tau_{2,2 }\tau_{3,4 }\tau_{4,1 } -
		\tau_{1,2 }\tau_{2,3 }\tau_{3,4 }\tau_{4,1 } - 
		\tau_{1,4 }\tau_{2,3 }\tau_{3,1 }\tau_{4,2} +
		\tau_{1,3 }\tau_{2,4 }\tau_{3,1 }\tau_{4,2 }+\\
		\tau_{1,4 }\tau_{2,1 }\tau_{3,3 }\tau_{4,2 } - 
		\tau_{1,1 }\tau_{2,4 }\tau_{3,3 }\tau_{4,2 } -
		\tau_{1,3 }\tau_{2,1 }\tau_{3,4 }\tau_{4,2} +
		\tau_{1,1 }\tau_{2,3 }\tau_{3,4 }\tau_{4,2} + \\
		\tau_{1,4 }\tau_{2,2 }\tau_{3,1 }\tau_{4,3 } -
		\tau_{1,2 }\tau_{2,4 }\tau_{3,1 }\tau_{4,3 } -
		\tau_{1,4 }\tau_{2,1 }\tau_{3,2 }\tau_{4,3} + 
		\tau_{1,1 }\tau_{2,4 }\tau_{3,2 }\tau_{4,3} +\\
		\tau_{1,2 }\tau_{2,1 }\tau_{3,4 }\tau_{4,3 } -
		\tau_{1,1 }\tau_{2,2 }\tau_{3,4 }\tau_{4,3 } - 
		\tau_{1,3 }\tau_{2,2 }\tau_{3,1 }\tau_{4,4} +
		\tau_{1,2 }\tau_{2,3 }\tau_{3,1 }\tau_{4,4} +\\
		\tau_{1,3 }\tau_{2,1 }\tau_{3,2 }\tau_{4,4 } - 
		\tau_{1,1 }\tau_{2,3 }\tau_{3,2 }\tau_{4,4 } -
		\tau_{1,2 }\tau_{2,1 }\tau_{3,3 }\tau_{4,4} +
		\tau_{1,1 }\tau_{2,2 }\tau_{3,3 }\tau_{4,4}.
	\end{gathered}
\end{equation}

	Let us find the largest root of the fourth degree equation from (\ref{poly4}). By the Ferrari method  \cite{Cardano},  \cite{Tabachnikov}, if $y_1$ is a root of the auxiliary equation
	\begin{equation}\label{poly4_sol1}
		\begin{gathered}
			y^3+Ay^2+By+C=0,
		\end{gathered}
	\end{equation}
	where $A=-b$, $B=ac-4d$, $C=-a^2d+4bd-c^2$,
	then the four roots of the original equation are found as the roots of two quadratic equations:
	\begin{equation}\label{poly4_sol2}
		\begin{gathered}
			x^2+(a/2)x+y_1/2=\\ \pm\sqrt{(a^2/4-b+y_1)x^2+((a/2)y_1-c)x+(y_1^2/4-d)},
		\end{gathered}
	\end{equation}
	where the radical expression on the right side is a perfect square.
.
	
	Using the Cardano-Vieta formulas \cite{Cardano}, \cite{Tabachnikov}, one of the roots of the equation (\ref{poly4_sol1}) is written as follows
	
	\begin{equation}\label{poly4_sol3}
		\begin{gathered}
			y_1=\frac{1}{3}(-A+2\sqrt{A^2-3B}\sin[1/3(\arcsin(\frac{2A^3-9AB+27C}{2(A^2-3B)^{3/2}})+2\pi)]),
		\end{gathered}
	\end{equation}

Then the root $\lambda_{\max}$ we need is the largest root of the equation
\begin{equation}\label{poly4_sol4}
	\begin{gathered}
		x^2+px+q=0,
	\end{gathered}
\end{equation}
where $p=a/2-\sqrt{a^2/4-b+y_1}$, $q=y_1/2+\sqrt{y_1^2/4-d}$.

\end{proof}

	\section{Exact solution for invariant percolation in the strip $2\times2\times \infty$}\label{sec_perc}

	Similarly, the partition function of the non-percolation configurations is found. The transfer matrix of non-percolation configurations $\theta'$ coincides with the transfer matrix of all model configurations, but with zero matrix elements corresponding to percolation configurations (Fig.3). 
	
		\begin{center}
		\begin{picture}(292,290) 
		\put(20,30){\line(0,1){252}}
		\put(20,282){\line(1,0){265}}
		\put(285,282){\line(0,-1){252}}
		\put(285,30){\line(-1,0){265}}
		\put(32,30){\line(0,1){204}}
		\put(44,30){\line(0,1){204}}
		\put(56,30){\line(0,1){204}}
		\put(68,30){\line(0,1){252}}
		\put(80,42){\line(0,1){204}}
		\put(92,42){\line(0,1){216}}
		\put(104,42){\line(0,1){204}}
		\put(116,42){\line(0,1){228}}
		\put(128,42){\line(0,1){204}}
		\put(140,42){\line(0,1){216}}
		\put(152,42){\line(0,1){204}}
		\put(164,42){\line(0,1){240}}
		\put(176,42){\line(0,1){204}}
		\put(188,42){\line(0,1){216}}
		\put(200,42){\line(0,1){204}}
		\put(212,42){\line(0,1){228}}
		\put(224,42){\line(0,1){204}}
		\put(236,42){\line(0,1){216}}
		\put(248,42){\line(0,1){204}}
		\put(260,42){\line(0,1){240}}
		\put(20,42){\line(1,0){265}}
		\put(260,54){\line(-1,0){204}}
		\put(260,66){\line(-1,0){216}}
		\put(260,78){\line(-1,0){204}}
		\put(260,90){\line(-1,0){228}}
		\put(260,102){\line(-1,0){204}}
		\put(260,114){\line(-1,0){216}}
		\put(260,126){\line(-1,0){204}}
		\put(260,138){\line(-1,0){240}}
		\put(260,150){\line(-1,0){204}}
		\put(260,162){\line(-1,0){216}}
		\put(260,174){\line(-1,0){204}}
		\put(260,186){\line(-1,0){228}}
		\put(260,198){\line(-1,0){204}}
		\put(260,210){\line(-1,0){216}}
		\put(260,222){\line(-1,0){204}}
		\put(285,234){\line(-1,0){265}}
		\put(285,246){\line(-1,0){217}}
		\put(285,258){\line(-1,0){217}}
		\put(285,270){\line(-1,0){217}}
		\put(58,34){\tiny{$\sigma_0^m$}}
		\put(46,34){\tiny{$\sigma_1^m$}}
		\put(34,34){\tiny{$\sigma_2^m$}}
		\put(22,34){\tiny{$\sigma_3^m$}}
		\put(58,46){\small{$-$}}
		\put(58,58){\small{$+$}}
		\put(58,70){\small{$-$}}
		\put(58,82){\small{$+$}}
		\put(58,94){\small{$-$}}
		\put(58,106){\small{$+$}}
		\put(58,118){\small{$-$}}
		\put(58,130){\small{$+$}}
		\put(58,142){\small{$-$}}
		\put(58,154){\small{$+$}}
		\put(58,166){\small{$-$}}
		\put(58,178){\small{$+$}}
		\put(58,190){\small{$-$}}
		\put(58,202){\small{$+$}}
		\put(58,214){\small{$-$}}
		\put(58,226){\small{$+$}}
		\put(46,52){\small{$-$}}
		\put(46,76){\small{$+$}}
		\put(46,100){\small{$-$}}
		\put(46,124){\small{$+$}}
		\put(46,148){\small{$-$}}
		\put(46,172){\small{$+$}}
		\put(46,196){\small{$-$}}
		\put(46,220){\small{$+$}}
		\put(34,64){\small{$-$}}
		\put(34,112){\small{$+$}}
		\put(34,160){\small{$-$}}
		\put(34,208){\small{$+$}}
		\put(22,88){\small{$-$}}
		\put(22,184){\small{$+$}}
		\put(262,238){\tiny{$\sigma_0^{m+1}$}}
		\put(262,250){\tiny{$\sigma_1^{m+1}$}}
		\put(262,262){\tiny{$\sigma_2^{m+1}$}}
		\put(262,274){\tiny{$\sigma_3^{m+1}$}}
			\put(250,238){\small{$-$}}
			\put(238,238){\small{$+$}}
			\put(226,238){\small{$-$}}
			\put(214,238){\small{$+$}}
			\put(202,238){\small{$-$}}
			\put(190,238){\small{$+$}}
			\put(178,238){\small{$-$}}
			\put(166,238){\small{$+$}}
			\put(154,238){\small{$-$}}
			\put(142,238){\small{$+$}}
			\put(130,238){\small{$-$}}
			\put(118,238){\small{$+$}}
			\put(106,238){\small{$-$}}
			\put(94,238){\small{$+$}}
			\put(82,238){\small{$-$}}
			\put(70,238){\small{$+$}}
			\put(244,250){\small{$-$}}
			\put(220,250){\small{$+$}}
			\put(196,250){\small{$-$}}
			\put(172,250){\small{$+$}}
			\put(148,250){\small{$-$}}
			\put(124,250){\small{$+$}}
			\put(100,250){\small{$-$}}
			\put(76,250){\small{$+$}}
			\put(232,262){\small{$-$}}
			\put(184,262){\small{$+$}}
			\put(136,262){\small{$-$}}
			\put(88,262){\small{$+$}}
			\put(208,274){\small{$-$}}
			\put(112,274){\small{$+$}}
			\put(72,225){\small{$*$}}
			\put(84,225){\small{$*$}}
			\put(96,225){\small{$*$}}
			\put(120,225){\small{$*$}}
			\put(132,225){\small{$*$}}
			\put(168,225){\small{$*$}}
			\put(192,225){\small{$*$}}
			\put(72,213){\small{$*$}}
			\put(84,213){\small{$*$}}
			\put(96,213){\small{$*$}}
			\put(120,213){\small{$*$}}
			\put(132,213){\small{$*$}}
			\put(168,213){\small{$*$}}
			\put(192,213){\small{$*$}}
			\put(72,201){\small{$*$}}
			\put(84,201){\small{$*$}}
			\put(96,201){\small{$*$}}
			\put(120,201){\small{$*$}}
			\put(132,201){\small{$*$}}
			\put(168,201){\small{$*$}}
			\put(192,201){\small{$*$}}
			\put(72,177){\small{$*$}}
			\put(84,177){\small{$*$}}
			\put(96,177){\small{$*$}}
			\put(120,177){\small{$*$}}
			\put(132,177){\small{$*$}}
			\put(168,177){\small{$*$}}
			\put(192,177){\small{$*$}}
			\put(72,165){\small{$*$}}
			\put(84,165){\small{$*$}}
			\put(96,165){\small{$*$}}
			\put(120,165){\small{$*$}}
			\put(132,165){\small{$*$}}
			\put(168,165){\small{$*$}}
			\put(192,165){\small{$*$}}
			\put(72,129){\small{$*$}}
			\put(84,129){\small{$*$}}
			\put(96,129){\small{$*$}}
			\put(120,129){\small{$*$}}
			\put(132,129){\small{$*$}}
			\put(168,129){\small{$*$}}
			\put(192,129){\small{$*$}}
			\put(72,105){\small{$*$}}
			\put(84,105){\small{$*$}}
			\put(96,105){\small{$*$}}
			\put(120,105){\small{$*$}}
			\put(132,105){\small{$*$}}
			\put(168,105){\small{$*$}}
			\put(192,105){\small{$*$}}
			\label{fig5}
		\end{picture}
		\\ FIGURE 3. The structure of the transfer matrix of $\theta'$ of non-percolation configurations. Non-null elements are marked with $*$ .
	\end{center}

    Further, it is shown that the largest eigenvalue $\mu_{\max}$ of the transfer matrix of non-percolation configurations $\theta'$ is the largest root of the characteristic equation of the matrix $\tau'$ with size $3\times3$:

\begin{equation}\label{tau'}
	\begin{gathered}	
		\tau' =\begin {psmallmatrix} 
		\begin{tabular}{p{1cm} p{3cm} p{1.5cm}} \tiny{
				\begin{center}$\theta_{1,1}$\end{center}}& \tiny{\begin{center}$\theta_{1,2}+\theta_{1,3}+\theta_{1,5}+ \theta_{1,9}$\end{center}}&\tiny{\begin{center} $\theta_{1,6}+\theta_{1,11}$\end{center}}\\
			\tiny{\begin{center}$\theta_{2,1}$\end{center}}& \tiny{\begin{center}$\theta_{2,2}+\theta_{2,3}+\theta_{2,5}+ \theta_{2,9}$\end{center}}&\tiny{\begin{center} $\theta_{2,6}+\theta_{2,11}$\end{center}}\\
			\tiny{\begin{center}$\theta_{6,1}$\end{center}}& \tiny{\begin{center}$\theta_{6,2}+\theta_{6,3}+\theta_{6,5}+ \theta_{6,9}$\end{center}}&\tiny{\begin{center} $\theta_{6,6}+\theta_{6,11}$\end{center}}\\
		\end{tabular}
	\end{psmallmatrix},\\	
\end{gathered}
\end{equation}
	 
	 \begin{theorem}\label{T2}
	 \textbf{Main percolation theorem}	
	  
	 The non-percolation probability $ P_M $ (\ref{no_percolation}) of a lattice model with a Hamiltonian (\ref{hamiltonian}) satisfies the following equality
	 \begin{equation}\label{Percolation1}
	 	\lim_{M\rightarrow\infty}{\frac{\ln{P_M}}{M}}=\ln{ \frac{\mu_{\max}}{\lambda_{\max}}},
	 \end{equation}
	 where $\mu_{\max}=\mu_{\max}(T)$ is the largest root is the largest root of the characteristic equation of the matrix $\tau'$ with size $3\times3$ (\ref{tau'}), and $\lambda_{\max}=\lambda_{\max}(T)$ is the largest root of the characteristic equation of the matrix $\tau$ (\ref{tau}).
	 
	 \end{theorem}
	 
	 \begin{proof}
	 
	 The proof is based on the commutation of the matrix $\theta'$ with the rotation matrix $D$ (\ref{D}). 
	 
	 Based on the structure of the new matrix $\theta'$, the eigenvector of the matrix $\theta'$ corresponding to the largest eigenvalue has the form:
	 \begin{displaymath}
	 	(x_{1} , x_{2} , x_{2} ,  0, x_{2} , x_{3}, 0, 0 , x_{2}, 0 ,  x_{3}, 0, 0, 0, 0, 0).
	 \end{displaymath}
	 
	The problem of finding the largest eigenvalue of the $\theta'$ matrix is reduced to finding the largest eigenvalue of the matrix $3\times3$ with size $\tau'$ (\ref{tau'}).

In this case the characteristic polynomial of the matrix $\tau'$ (\ref{tau'}) is equal to:
\begin{equation}\label{char_poly'}
	\begin{gathered}
		\mu^3+a'\mu^2+b'\mu+c'=0
	\end{gathered}
\end{equation}

where

\begin{equation}\label{a'}
	\begin{gathered}
		a'=-
		(\tau'_{1,1 }+
		\tau'_{2,2 }+
		\tau'_{3,3 })
	\end{gathered}
\end{equation}
\begin{equation}\label{b'}
	\begin{gathered}
		b'=
		-\tau'_{1,2 }\tau'_{2,1 }+ 
		\tau'_{1,1 }\tau'_{2,2 } -
		\tau'_{1,3 }\tau'_{3,1 } -
		\tau'_{2,3 }\tau'_{3,2 } +
		\tau'_{1,1 }\tau'_{3,3 } +
		\tau'_{2,2 }\tau'_{3,3 }
	\end{gathered}
\end{equation}
\begin{equation}\label{c'}
	\begin{gathered}
		c'=\tau'_{1,3 }\tau'_{2,2 }\tau'_{3,1 } -
		\tau'_{1,2 }\tau'_{2,3 }\tau'_{3,1 } -
		\tau'_{1,3 }\tau'_{2,1 }\tau'_{3,2 } +
		\tau'_{1,1 }\tau'_{2,3 }\tau'_{3,2 } +
		\tau'_{1,2 }\tau'_{2,1 }\tau'_{3,3 } - \\
		\tau'_{1,1 }\tau'_{2,2 }\tau'_{3,3 } 
	\end{gathered}
\end{equation}

The largest root of the equation(\ref{char_poly'}) is:	
\begin{equation}\label{poly3}
	\begin{gathered}
		\mu_{\max}=\frac{1}{3}(-a'+2\sqrt{(a')^2-3b'}\sin[1/3(\arcsin(\frac{2(a')^3-9a'b'+27c'}{2((a')^2-3b')^{3/2}})+2\pi)]),
	\end{gathered}
\end{equation}

Then the non-percolation probability (\ref{no_percolation}) is rewritten as follows
\begin{displaymath}
	\frac{1}{M}\ln{P_M}=\frac{1}{M}\ln{\frac{Z'}{Z}}=\frac{1}{M}\ln{\frac{\mu_{\max}^M}{\lambda_{\max}^M}}\underset{M\rightarrow\infty}{\longrightarrow} \ln{ \frac{\mu_{\max}}{\lambda_{\max}}}.
\end{displaymath}

\end{proof}

\section{Some simplest particular cases of solution }\label{Simplest cases}

\subsection{The first particular case}
Now let us consider the Hamiltonian, with non-zero generating interactions (other interactions are zero)
\begin{displaymath}
	\begin{gathered}
	J_{\{t_0^m,t_0^{m+1}\}}=J_1, J_{\{t_0^m,t_1^m,t_2^m,t_3^m\}}=J_2, J_{\{t_0^{m+1},t_1^{m+1},t_2^{m+1},t_3^{m+1}\}}=J_2,\\ J_{\{t_0^{m},t_1^{m},t_2^{m},t_3^{m+1}\}}=J_3, J_{\{t_3^{m},t_0^{m+1},t_1^{m+1},t_2^{m+1}\}}=J_4,
	J_{\{t_1^{m},t_2^{m},t_3^{m},t_1^{m+1},t_2^{m+1},t_3^{m+1}\}}=J_5,\\
	J_{\{t_0^{m},t_1^{m},t_2^{m},t_3^{m},t_0^{m+1},t_1^{m+1},t_2^{m+1},t_3^{m+1}\}}=J_6.
    \end{gathered}
\end{displaymath}
	The generating Hamiltonian of the model has the form

	\begin{equation}\label{hamiltonian_part1}
	\begin{gathered} 
		\mathcal{\hat{H}}^m=J_1\sigma_0^m\sigma_0^{m+1}+J_2(\sigma_0^m\sigma_1^{m}\sigma_2^{m}\sigma_3^{m}+\sigma_0^{m+1}\sigma_1^{m+1}\sigma_2^{m+1}\sigma_3^{m+1})+\\
		J_3\sigma_0^m\sigma_1^{m}\sigma_2^{m}\sigma_3^{m+1}+J_4\sigma_0^{m}\sigma_1^{m+1}\sigma_2^{m+1}\sigma_3^{m+1}+J_5\sigma_1^{m}\sigma_2^{m}\sigma_3^{m}\sigma_1^{m+1}\sigma_2^{m+1}\sigma_3^{m+1}+\\J_6\sigma_0^m\sigma_1^{m}\sigma_2^{m}\sigma_3^{m}\sigma_0^{m+1}\sigma_1^{m+1}\sigma_2^{m+1}\sigma_3^{m+1}.
	\end{gathered} 	
\end{equation}

	Further, for convenience, we introduce the notation $p, q, u, v$ with lower indexing, where $p$ corresponds to the interaction coefficient of two spins, $q$ - four spins, $u$ - six spins, $v$ - eight spins.

 Let $p_1=\exp(J_{\{t_0^m,t_0^{m+1}\}}), q_1=\exp(8J_{\{t_0^m,t_1^{m},t_2^{m},t_3^{m}\}}), q_2=\exp(J_{\{t_0^m,t_1^{m},t_2^{m},t_3^{m+1}\}}),\\ q_{3}=\exp(J_{\{t_0^m,t_1^{m+1},t_2^{m+1},t_3^{m+1}\}}), u_1=\exp(J_{\{t_0^m,t_1^{m},t_2^{m},t_0^{m+1},t_1^{m+1},t_2^{m+1}\}}),\\ v_{1}=\exp(4J_{\{t_0^m,t_1^{m},t_2^{m},t_3^{m},t_0^{m+1},t_1^{m+1},t_2^{m+1},t_3^{m+1}\}})$ . Then the following theorem holds.
\begin{theorem}\label{T3} 
	In the thermodynamic limit, the free energy of the model with the generating Hamiltonian (\ref{hamiltonian_part1}) and Hamiltonian (\ref{hamiltonian},\ref{H_m}) can be represented as (\ref{free_energy_lambda}), internal energy as (\ref{int_energy_lambda}), heat capacity as (\ref{heat_capacity_lambda}), while
	
	\begin{equation}\label{lambda_max_part3}
		\begin{gathered}
			\lambda_{\max}=(-a+\sqrt{a^2-4b})/2,
		\end{gathered}
	\end{equation}
  where 
  \begin{equation}\label{a_part3}
  	\begin{gathered}
  		a=-6 v_1q_1^{-1} - 6 q_1 v_1 - q_1 v_1p_1^{-4} q_2^{-4} q_{3}^{-4} u_1^{-4} - 
  		q_2^4 q_{3}^4 v_1 p_2^{-4} q_1^{-1} u_2^{-4} -\\ p_2^4 u_2^4 v_1 q_1^{-1} q_2^{-4} q_{3}^{-4} - 
  		p_1^4 q_1 q_{2}^4 q_{3}^4 u_{1}^4 v_{1},
  	\end{gathered}
  \end{equation}
  \begin{equation}\label{b_part3}
  	\begin{gathered}
  		b=-16 q_{3}^4 q_2^{-4} v_1^{-2} - 16 q_2^4 q_{3}^{-4} v_1^{-2} - 16 p_1^4
  		u_1^{-4} v_1^{-2} - 16 u_1^4 p_1^{-4} v_1^{-2} + 36 v_1^2 +\\
  		v_1^2 q_2^{-8} q_{3}^{-8} + 
  		q_{3}^8 q_{5}^8 v_{1}^2 + 
  		v_1^2 p_1^{-8} u_1^{-8} + 
  		6 v_1^2 p_{2}^{-4} q_{5}^{-4} q_{3}^{-4} u_1^{-4} + 
  		6 q_{5}^4 q_{3}^4 v_{1}^2 p_{2}^{-4} u_1^{-4} +\\ 
  		6 p_1^4 u_1^4 v_1^2  q_2^{-4} q_{3}^{-4} + 6 p_1^4 q_2^4 q_{3}^4 u_1^4 v_1^2 + 
  		p_1^8 u_1^8 v_1^2.
  	\end{gathered}
  \end{equation}

\end{theorem}

\begin{proof}

The choice of these nonzero interactions is justified by the fact that the transfer matrix $\theta$ corresponding to such a Hamiltonian (\ref{hamiltonian_part1}) commutes with the matrix $\Phi$:

\begin{equation}\label{Phi}
	  \Phi=
	\begin {psmallmatrix} 
	0& 0& 0& 1& 0& 0& 0& 0& 0& 0& 0& 0& 0& 0& 0& 0 \\
	0& 0& 1& 0& 0& 0& 0& 0& 0& 0& 0& 0& 0& 0& 0& 0 \\
	0& 1& 0& 0& 0& 0& 0& 0& 0& 0& 0& 0& 0& 0& 0& 0 \\
	1& 0& 0& 0& 0& 0& 0& 0& 0& 0& 0& 0& 0& 0& 0& 0 \\
	0& 0& 0& 0& 0& 0& 0& 1& 0& 0& 0& 0& 0& 0& 0& 0 \\
	0& 0& 0& 0& 0& 0& 1& 0& 0& 0& 0& 0& 0& 0& 0& 0 \\
	0& 0& 0& 0& 0& 1& 0& 0& 0& 0& 0& 0& 0& 0& 0& 0 \\
	0& 0& 0& 0& 1& 0& 0& 0& 0& 0& 0& 0& 0& 0& 0& 0 \\
	0& 0& 0& 0& 0& 0& 0& 0& 0& 0& 0& 1& 0& 0& 0& 0 \\
	0& 0& 0& 0& 0& 0& 0& 0& 0& 0& 1& 0& 0& 0& 0& 0 \\
	0& 0& 0& 0& 0& 0& 0& 0& 0& 1& 0& 0& 0& 0& 0& 0 \\
	0& 0& 0& 0& 0& 0& 0& 0& 1& 0& 0& 0& 0& 0& 0& 0 \\
	0& 0& 0& 0& 0& 0& 0& 0& 0& 0& 0& 0& 0& 0& 0& 1 \\
	0& 0& 0& 0& 0& 0& 0& 0& 0& 0& 0& 0& 0& 0& 1& 0 \\
	0& 0& 0& 0& 0& 0& 0& 0& 0& 0& 0& 0& 0& 1& 0& 0 \\
	0& 0& 0& 0& 0& 0& 0& 0& 0& 0& 0& 0& 1& 0& 0& 0 \\
\end{psmallmatrix} 
\end{equation}

	Let us write out the eigenvectors of the matrix $\Phi$ (\ref{Phi}) corresponding to the eigenvalue $1$:
\begin{equation}
	\begin{gathered}
		(0, 0, 0, 0, 0, 0, 0, 0, 0, 0, 0, 0, 1, 0, 0, 1),\\
		(0, 0, 0, 0, 0, 0, 0, 0, 0, 0, 0, 0, 0, 1, 1, 0),\\
		(0, 0, 0, 0, 0, 0, 0, 0, 1, 0, 0, 1, 0, 0, 0, 0),\\
		(0, 0, 0, 0, 0, 0, 0, 0, 0, 1, 1, 0, 0, 0, 0, 0),\\
		(0, 0, 0, 0, 1, 0, 0, 1, 0, 0, 0, 0, 0, 0, 0, 0),\\
		(0, 0, 0, 0, 0, 1, 1, 0, 0, 0, 0, 0, 0, 0, 0, 0)
		,\\
		(1, 0, 0, 1, 0, 0, 0, 0, 0, 0, 0, 0, 0, 0, 0, 0)
		,\\
		(0, 1, 1, 0, 0, 0, 0, 0, 0, 0, 0, 0, 0, 0, 0, 0).
	\end{gathered}
\end{equation}

Then the eigenvector of the matrix $\tau$ corresponding to the largest eigenvalue has the form:
\begin{displaymath}
	(x_{1} , x_{2} , x_{2}, x_{1} , x_{3} , x_{4} , x_{4} , x_{3} , x_{3}, x_{4}  , x_{4} , x_{3} , x_{1} , x_{2}, 
	x_{2} , x_{1}).
\end{displaymath}

Taking into account the commutation of the transfer matrix (\ref{teta_k,l}) for the considered model with the matrix $D$ (\ref{D}), we find the eigenvector in the form:
\begin{displaymath}
	(x_{1} , x_{2} , x_{2}, x_{1} , x_{2} , x_{1} , x_{1} , x_{2} , x_{2}, x_{1}  , x_{1} , x_{2} , x_{1} , x_{2}, x_{2} , x_{1}).
\end{displaymath}

This kind of eigenvector allows us to reduce the search for the largest eigenvalue of the matrix $\theta$ to the matrix $\tau $ with size $2\times2$ with elements:

		 \begin{equation}\label{tau_part1}
			\begin{gathered}		
				\tau =\begin {psmallmatrix} 
				\begin{tabular}{p{3cm} p{3cm}} \tiny{
						\begin{center}$\theta_{1,1}+\theta_{1,4}+\theta_{1,6}+\theta_{1,7}+\theta_{1,10}+\theta_{1,11}+\theta_{1,13}+\theta_{1,16}$ \end{center}}& \tiny{\begin{center}$\theta_{1,2}+\theta_{1,3}+\theta_{1,5}+\theta_{1,8}+\theta_{1,9}+\theta_{1,12}+\theta_{1,14}+\theta_{1,15}$ \end{center}}\\
					\tiny{
						\begin{center}$\theta_{2,1}+\theta_{2,4}+\theta_{2,6}+\theta_{2,7}+\theta_{2,10}+\theta_{2,11}+\theta_{2,13}+\theta_{2,16}$ \end{center}}& \tiny{\begin{center}$\theta_{2,2}+\theta_{2,3}+\theta_{2,5}+\theta_{2,8}+\theta_{2,9}+\theta_{2,12}+\theta_{2,14}+\theta_{2,15}$ \end{center}}\\
				\end{tabular}
			\end{psmallmatrix}.\\	
		\end{gathered}
	\end{equation}

 The characteristic polynomial of matrix $\tau$ (\ref{tau_part1}) is:

\begin{equation}\label{char_poly_part3}
	\begin{gathered}
		\lambda^2+a\lambda+b=0.
	\end{gathered}
\end{equation}

Due to the simplicity of the model, the coefficients of the characteristic polynomial are written explicitly (\ref{a_part3}, \ref{b_part3}), the largest root of the equation is determined by the formula (\ref{lambda_max_part3}).

The polynomial containing the largest eigenvalue of the transfer matrix of non-percolation configurations will be of the third degree and will have the form (\ref{char_poly'}), polynomial coefficients are determined by the formulas (\ref{a'}-\ref{c'}), and the largest eigenvalue of the transfer matrix of non-percolation configurations (\ref{poly3}).

\end{proof}

\subsection{The second particular case}

Let us consider one more special case, when only interactions remain nonzero

 \begin{displaymath}
 	\begin{gathered}
 		J_{\{t_0^m,t_0^{m+1}\}}=J_1, J_{\{t_0^m,t_1^m,t_0^{m+1},t_1^{m+1}\}}=J_2,
 		J_{\{t_0^m,t_2^m,t_0^{m+1},t_2^{m+1}\}}=J_3, \\
 		J_{\{t_1^{m},t_2^{m},t_3^{m},t_1^{m+1},t_2^{m+1},t_3^{m+1}\}}=J_4,
 		J_{\{t_0^{m},t_1^{m},t_2^{m},t_3^{m},t_0^{m+1},t_1^{m+1},t_2^{m+1},t_3^{m+1}\}}=J_5.
 	\end{gathered}
 \end{displaymath}

The generating Hamiltonian of the model has the form:

	\begin{equation}\label{hamiltonian_part2}
	\begin{gathered} 
		\mathcal{\hat{H}}^m=J_1\sigma_0^m\sigma_0^{m+1}+J_2\sigma_0^m\sigma_1^{m}\sigma_0^{m+1}\sigma_1^{m+1}+J_3\sigma_0^m\sigma_2^{m}\sigma_0^{m+1}\sigma_2^{m+1}+\\
		J_4\sigma_1^{m}\sigma_2^{m}\sigma_3^{m}\sigma_1^{m+1}\sigma_2^{m+1}\sigma_3^{m+1}+\\J_5\sigma_0^m\sigma_1^{m}\sigma_2^{m}\sigma_3^{m}\sigma_0^{m+1}\sigma_1^{m+1}\sigma_2^{m+1}\sigma_3^{m+1},	
	\end{gathered} 	
\end{equation}

 Let  $p_1=\exp(J_{\{t_0^m,t_0^{m+1}\}}), 
  q_1=\exp(J_{\{t_0^m,t_1^{m},t_0^{m+1},t_1^{m+1}\}}), q_{2}=
  \exp(J_{\{t_0^m,t_2^{m},t_0^{m+1},t_2^{m+1}\}}),\\ u_1=\exp(J_{\{t_0^m,t_1^{m},t_2^{m},t_0^{m+1},t_1^{m+1},t_2^{m+1}\}}),\\ v_{1}=\exp(4J_{\{t_0^m,t_1^{m},t_2^{m},t_3^{m},t_0^{m+1},t_1^{m+1},t_2^{m+1},t_3^{m+1}\}})$ .

\begin{theorem}\label{T4} 
	In the thermodynamic limit, the free energy of the model with the generating Hamiltonian (\ref{hamiltonian_part2}) and the Hamiltonian (\ref{hamiltonian}-\ref{H_m}) can be represented as (\ref{free_energy_lambda}), the internal energy as (\ref{int_energy_lambda}), heat capacity in the form (\ref{heat_capacity_lambda}), while
	\begin{equation}\label{lambda_part2}
		\begin{gathered}
			\lambda_{\max}=4 p_1^2 u_1^{-2} v_1^{-1} + 4 u_1^2 p_1^{-2} v_1^{-1} + 4 v_1 q_{2}^{-2} + 2 q_{2}^2 v_1 q_1^{-4} + q_1^4 q_{2}^2 v_1 p_1^{-4} u_1^{-4} + \\
			p_1^4 q_1^4 q_{2}^2 u_1^4 v_1.
	\end{gathered}
\end{equation}

\end{theorem}

\begin{proof}

 The transfer matrix of such an interaction commutes with the matrix $\Psi$:

\begin{equation}\label{Psi}
	 \Psi=
	\begin {psmallmatrix} 
	0& 1& 0& 0& 0& 0& 0& 0& 0& 0& 0& 0& 0& 0& 0& 0 \\
	1& 0& 0& 0& 0& 0& 0& 0& 0& 0& 0& 0& 0& 0& 0& 0 \\
	0& 0& 0& 1& 0& 0& 0& 0& 0& 0& 0& 0& 0& 0& 0& 0 \\
	0& 0& 1& 0& 0& 0& 0& 0& 0& 0& 0& 0& 0& 0& 0& 0 \\
	0& 0& 0& 0& 0& 1& 0& 0& 0& 0& 0& 0& 0& 0& 0& 0 \\
	0& 0& 0& 0& 1& 0& 0& 0& 0& 0& 0& 0& 0& 0& 0& 0 \\
	0& 0& 0& 0& 0& 0& 0& 1& 0& 0& 0& 0& 0& 0& 0& 0 \\
	0& 0& 0& 0& 0& 0& 1& 0& 0& 0& 0& 0& 0& 0& 0& 0 \\
	0& 0& 0& 0& 0& 0& 0& 0& 0& 1& 0& 0& 0& 0& 0& 0 \\
	0& 0& 0& 0& 0& 0& 0& 0& 1& 0& 0& 0& 0& 0& 0& 0 \\
	0& 0& 0& 0& 0& 0& 0& 0& 0& 0& 0& 1& 0& 0& 0& 0 \\
	0& 0& 0& 0& 0& 0& 0& 0& 0& 0& 1& 0& 0& 0& 0& 0 \\
	0& 0& 0& 0& 0& 0& 0& 0& 0& 0& 0& 0& 0& 1& 0& 0 \\
	0& 0& 0& 0& 0& 0& 0& 0& 0& 0& 0& 0& 1& 0& 0& 0 \\
	0& 0& 0& 0& 0& 0& 0& 0& 0& 0& 0& 0& 0& 0& 0& 1 \\
	0& 0& 0& 0& 0& 0& 0& 0& 0& 0& 0& 0& 0& 0& 1& 0 \\
\end{psmallmatrix} 
\end{equation}

Let us write out the eigenvectors of the matrix $\Psi$ (\ref{Psi}) corresponding to the eigenvalue $1$:
\begin{equation}
	\begin{gathered}
		(1, 1, 0, 0, 0, 0, 0, 0, 0, 0, 0, 0, 0, 0, 0, 0),\\
		(0, 0, 1, 1, 0, 0, 0, 0, 0, 0, 0, 0, 0, 0, 0, 0),\\
		(0, 0, 0, 0, 1, 1, 0, 0, 0, 0, 0, 0, 0, 0, 0, 0),\\
		(0, 0, 0, 0, 0, 0, 1, 1, 0, 0, 0, 0, 0, 0, 0, 0),\\
		(0, 0, 0, 0, 0, 0, 0, 0, 1, 1, 0, 0, 0, 0, 0, 0),\\
		(0, 0, 0, 0, 0, 0, 0, 0, 0, 0, 1, 1, 0, 0, 0, 0)
		,\\
		(0, 0, 0, 0, 0, 0, 0, 0, 0, 0, 0, 0, 1, 1, 0, 0)
		,\\
		(0, 0, 0, 0, 0, 0, 0, 0, 0, 0, 0, 0, 0, 0, 1, 1).
	\end{gathered}
\end{equation}

Then the eigenvector of the matrix $\theta$, corresponding to the largest eigenvalue, can be found in the form:
\begin{displaymath}
	(x_{1} , x_{1} , x_{2}, x_{2} , x_{3} , x_{3} , x_{4} , x_{4} , x_{4}, x_{4}  , x_{3} , x_{3} , x_{2} , x_{2}, 
	x_{1} , x_{1}).
\end{displaymath}

Taking into account the commutation with the matrix $D$ (\ref{d}), the eigenvector corresponding to the largest eigenvalue has a one-component form 
\begin{displaymath}
	(x_{1} , x_{1} , x_{1}, x_{1} , x_{1} , x_{1} , x_{1} , x_{1} , x_{1}, x_{1}  , x_{1} , x_{1} , x_{1} , x_{1}, x_{1} , x_{1}).
\end{displaymath}

Then the largest eigenvalue	$\lambda_{\max}$ has the form (\ref{lambda_part2}).

As in the previous special case, the degree of the polynomial containing the largest eigenvalue for tight configurations is not reduced (\ref{char_poly'}), polynomial coefficients are determined by the formulas (\ref{a'}-\ref{c'}), and the largest eigenvalue of the transfer matrix of non-percolation configurations according to the formula (\ref{poly3}).

\end{proof}

\section{The Gonihedric Ising Model}\label{gonihedric}

\subsection{Exact solution of the gonigendrial model with free boundary conditions}

Let us consider a gonigendrial model \cite{Ambartzumian},\cite{Johnston_1996}, \cite{Pelizzola} or a model with the interaction of nearest neighbors, next nearest neighbors and plaquettes, and
\begin{displaymath}
	\begin{gathered}
2 J_{\{t_0^m,t_1^m\}}= 2 J_{\{t_0^{m+1},t_1^{m+1}\}}= J_{\{t_0^{m},t_0^{m+1}\}}=J_1, \\ 4 J_{\{t_0^m,t_2^m\}}= 4 J_{\{t_0^{m+1},t_2^{m+1}\}}=J_{\{t_0^{m},t_1^{m+1}\}}= J_{\{t_0^{m},t_3^{m+1}\}}=J_2,\\ 8 J_{\{t_0^m,t_1^m,t_2^m,t_3^m\}}= 8 J_{\{t_0^{m+1},t_1^{m+1},t_2^{m+1},t_3^{m+1}\}}= J_{\{t_0^{m},t_1^{m},t_0^{m+1},t_1^{m+1}\}}=J_3,
\end{gathered}
\end{displaymath}
and the remaining coefficients of the Hamiltonian (\ref{hamiltonian}) are zero.

The classical notation of this Hamiltonian \cite{Johnston_1996} for the considering model is
\begin{equation}\label{hamiltonian_part_gonihedric}
	\begin{gathered} 
		\mathcal H(\sigma)= -\sum_{m=0}^{M-1}(
		J_{1}((\sigma_0^m\sigma_1^m+\sigma_1^m\sigma_2^m+\sigma_2^m\sigma_3^m+\sigma_3^m\sigma_0^m+\sigma_0^{m+1}\sigma_1^{m+1}+\\
		\sigma_1^{m+1}\sigma_2^{m+1}+\sigma_2^{m+1}\sigma_3^{m+1}+\sigma_3^{m+1}\sigma_0^{m+1})/2+\sigma_0^m\sigma_0^{m+1}+\sigma_1^m\sigma_1^{m+1}+\\
		\sigma_2^m\sigma_2^{m+1}+\sigma_3^m\sigma_3^{m+1}) + J_{2}((\sigma_0^m\sigma_2^m+\sigma_1^m\sigma_3^m+\sigma_0^{m+1}\sigma_2^{m+1}+\\
		\sigma_1^{m+1}\sigma_3^{m+1})/2+\sigma_0^m\sigma_1^{m+1}+\sigma_1^m\sigma_2^{m+1}+
		\sigma_2^m\sigma_3^{m+1}+\sigma_3^m\sigma_0^{m+1}+\sigma_0^m\sigma_3^{m+1}+\\
		\sigma_1^m\sigma_0^{m+1}+
		\sigma_2^m\sigma_1^{m+1}+\sigma_3^m\sigma_2^{m+1}) +J_{4}((\sigma_0^m\sigma_1^{m}\sigma_2^{m}\sigma_3^{m}+\\
		\sigma_0^{m+1}\sigma_1^{m+1}\sigma_2^{m+1}\sigma_3^{m+1})/2+\sigma_0^m\sigma_1^{m}\sigma_0^{m+1}\sigma_1^{m+1}+\sigma_1^m\sigma_2^{m}\sigma_1^{m+1}\sigma_2^{m+1}+\\
		\sigma_2^m\sigma_3^{m}\sigma_2^{m+1}\sigma_3^{m+1}+\sigma_3^m\sigma_0^{m}\sigma_3^{m+1}\sigma_0^{m+1})).			 
	\end{gathered} 	
\end{equation}

Let us assume that  $p_1=\exp(J_{\{t_0^m,t_0^{m+1}\}}), 
p_2=\exp(J_{\{t_0^m,t_1^{m+1}\}}), \\
q_1=\exp(J_{\{t_0^m,t_1^{m},t_0^{m+1},t_1^{m+1}\}})$ .

Then the matrices $\tau$ (\ref{tau}) and  $\tau'$ (\ref{tau'}) take the form:
\begin{equation}\label{t_gon}
	\begin{gathered}		
	\tau =\begin {psmallmatrix} 
	\begin{tabular}{p{1.5cm} p{3cm} p{2.5cm}  p{1.5cm}} \tiny{
			\begin{center}$q_1^5 p_2^{-6} + p_1^8 p_2^{10} q_1^5$ \end{center}}& \tiny{\begin{center}$4 p_2^{-3} + 4 p_1^4 p_2^5$ \end{center}}& \tiny{\begin{center}$4 p_1^2 q_1$ \end{center}}&\tiny{\begin{center} $2 p_2^2 q_1^{-3}$\end{center}}\\
		\tiny{\begin{center}$p_2^{-3} + p_1^4 p_2^5$ \end{center}}&\tiny{ \begin{center}$2 q_1^{-5} + 2 p_2^{-4} q_1^{-1} + 2 p_2^4 q_1^{-1} + q_1^3 p_1^{-4} + p_1^4 q_1^3$ \end{center}}&\tiny{ \begin{center}$2 p_1^{-2} p_2^{-1} + 2 p_1^2 p_2^{-1}$ \end{center}}&\tiny{\begin{center} $p_2^{-3} + p_2^5 p_1^{-4}$\end{center}}\\
		\tiny{\begin{center}$2 p_1^2 q_1$ \end{center}}&\tiny{ \begin{center}$4 p_1^{-2} p_2^{-1} + 4 p_1^2 p_2^{-1}$ \end{center}}&\tiny{ \begin{center}$2 p_2^{-2} q_1^{-3} + q_1^5 p_1^{-4} p_2^{-2} + p_1^4 q_1^5 p_2^{-2}$ \end{center}}&\tiny{\begin{center} $2 q_1 p_1^{-2}$\end{center}}\\
		\tiny{\begin{center}$2 p_2 q_1^{-2}$ \end{center}}& \tiny{\begin{center}$4 p_2^{-3} + 4 p_2^5 p_1^{-4}$ \end{center}}& \tiny{\begin{center}$4 q_1 p_1^{-2}$ \end{center}}&\tiny{\begin{center} $q_1^5 p_2^{-6} + p_2^{10} q_1^5 p_1^{-8}$\end{center}}
	\end{tabular}
\end{psmallmatrix},\\	
\end{gathered}
\end{equation}

\begin{equation}\label{t'_gon}
	\begin{gathered}		
		\tau' =\begin {psmallmatrix} 
		\begin{tabular}{p{1.5cm} p{3cm} p{2.5cm}} \tiny{
				\begin{center}$p_1^8 p_2^{10} q_1^5$ \end{center}}& \tiny{\begin{center}$4 p_1^4 p_2^5$ \end{center}}& \tiny{\begin{center}$2 p_2^2 q_1^{-3}$ \end{center}}\\
			\tiny{\begin{center}$p_1^4 p_2^5$ \end{center}}&\tiny{ \begin{center}$q_1^{-5} + 2 p_2^4 q_1^{-1} + p_1^4 q_1^3$ \end{center}}&\tiny{ \begin{center}$p_2^{-3} + p_2^5 p_1^{-4}$ \end{center}}\\
			\tiny{\begin{center}$p_2 q_1^{-2}$ \end{center}}&\tiny{ \begin{center}$2 p_2^{-3} + 2 p_2^5 p_1^{-4}$ \end{center}}&\tiny{ \begin{center}$q_1^5 p_2^{-6} + p_2^{10} q_1^5 p_1^{-8}$\end{center}}\\
		\end{tabular}
	\end{psmallmatrix}.\\	
\end{gathered}
\end{equation}

	Now we use Theorem \ref{T1} and Theorem \ref{T2}, and write down as their corollary the results for the gonigendrial model
 \begin{theorem}\label{S1}
 	
 	In the thermodynamic limit, the free energy of the gonigendrial model with the Hamiltonian (\ref{hamiltonian_part_gonihedric}) is represented as (\ref{free_energy_lambda}), the internal energy as (\ref{int_energy_lambda}), the heat capacity as (\ref{heat_capacity_lambda}), where $\lambda_{\max}$ is the largest root of the characteristic equation of the fourth power of the matrix $\tau$ (\ref{t_gon}), and the relation (\ref{Percolation1}) is true for the non-percolation probability, where $\mu_{\max}$ is the largest root of the characteristic equation of the matrix $\tau'$ (\ref{t'_gon}).

\end{theorem}

The coefficients of the characteristic polynomials for $\tau$ (\ref{t_gon}) and $\tau'$ (\ref{t'_gon}) are written out in the formulas (\ref{a})-(\ref{d}) и (\ref{a'})-(\ref{c'}). $\lambda_{max}$ is found according to the formulas (\ref{poly4}-\ref{poly4_sol4}), and $\mu_{max}$ is found according to the formula (\ref{poly3}).

In the classical setting \cite{Ambartzumian} the coefficients of the gonigendrial model are related by the relations $J_2=rJ_1$, $J_3=4k/(1-k)$. For example, \cite{Pelizzola} obtained phase diagrams for $k=0$ and $k=1/3$. Let us consider similar cases (example \ref{ex1} and example \ref{ex2}).

\begin{primer} \label{ex1} Let us find the thermodynamic characteristics of the gonigendrial model for $J_{1}=1$, $J_{2}=r$, $J_{3}=0$, $T\in[0.1,5]$, $r\in[- 2,2]$ in the absence of an external magnetic field.  One can see the appearance of phase transition lines beginning at $T>0$, $r=-1$ and $r=-0.25$ on the heat capacity graph (Fig.4).

\begin{center}
\begin{figure}[h]
	\includegraphics[scale=1]{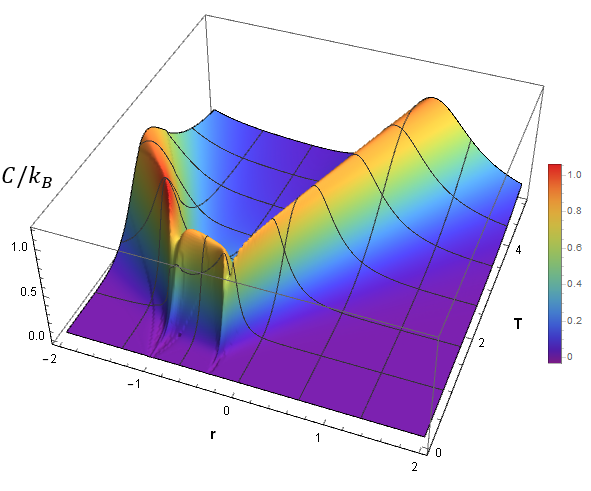}\\
	FIGURE 4. The heat capacity graph for the example \ref{ex1}.
\end{figure}
\end{center}

Large gradients of the impermeability probability (Fig. 5) repeat the lines of incipient phase transitions of the heat capacity.

\begin{center}
\begin{figure}[h]
	\includegraphics[scale=0.5]{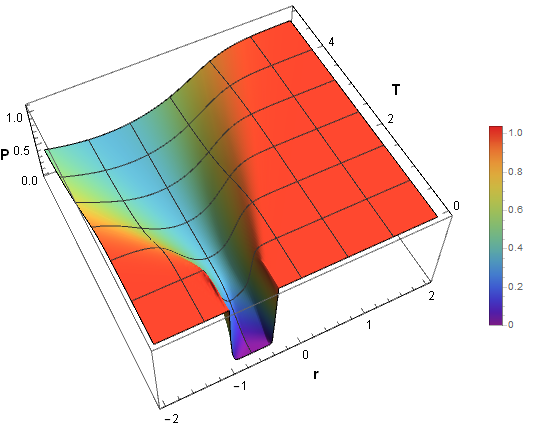}\\
	FIGURE 5. The probability of non-percolation for the example \ref{ex1}.
\end{figure}
\end{center}

\end{primer}

\begin{primer} \label{ex2} Under the conditions of the previous example, let us consider the case of a nonzero plaquette interaction $J_{3}=1/2$, $T\in[0.1,5]$, $r\in[-2,2]$. One can see the appearance of phase transition lines beginning at $T>0$, $r=-1$ and $r=-0.25$ on the heat capacity graph.

At $T\rightarrow 0 $, the heat capacity graph (Fig.6) shows a longer break at the points $r=-1$ and $r=-0.25$ than the one observed in the \ref{ex1} example in Fig. 4.

\begin{center}
\begin{figure}[h]
	\includegraphics[scale=0.5]{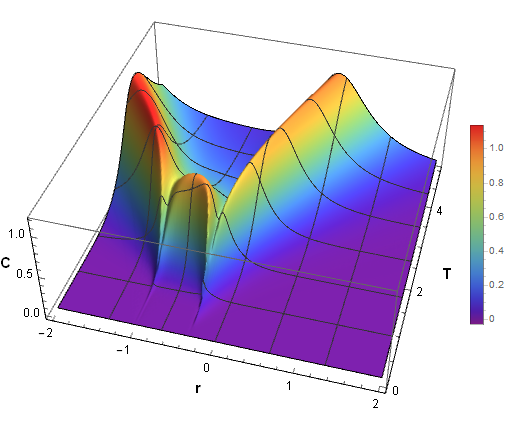}\\
	FIGURE 6. The heat capacity graph for the example \ref{ex2}.
\end{figure}
\end{center}

\end{primer}

\subsection{Numerical solution of the gonigendrial model $3\times3\times\infty$}\label{gon_3x3x3}
 The authors of \cite{Pelizzola} obtained a phase plane, where the phase transition line corresponds to the value $r=-0.25$. Thus, the shape of the phase plane differs from the picture obtained for the strip $2\times2\times\infty$. In this regard, the authors assumed that the model with the $2\times2$ section is strongly subject to boundary conditions, so the gonigendrial model was studied numerically $3\times3\times\infty$.
  
 In the case of free boundary conditions with the Hamiltonian (\ref{hamiltonian_part_gonihedric_3x3xinf_open}) we have:
  
  \begin{multline}\label{hamiltonian_part_gonihedric_3x3xinf_open}
  		\mathcal H(\sigma)= -\sum_{m=0}^{M-1}(
  		J_{1}(\frac{1}{2}\sum_{i=0}^{2}\sum_{j=0}^{1}(\sigma_{i,j}^m\sigma_{i,j+1}^m+\sigma_{j,i}^m\sigma_{j+1,i}^m+\sigma_{i,j}^{m+1}\sigma_{i,j+1}^{m+1}+\\\sigma_{j,i}^{m+1}\sigma_{j+1,i}^{m+1})+
  		\sum_{i=0}^{2}\sum_{j=0}^{2}\sigma_{i,j}^m\sigma_{i,j}^{m+1})
  		+J_2(\frac{1}{2}\sum_{i=0}^{1}\sum_{j=0}^{1}(\sigma_{i,j}^m\sigma_{i+1,j+1}^m+\\\sigma_{i,j+1}^m\sigma_{i+1,j}^m+\sigma_{i,j}^{m+1}\sigma_{i+1,j+1}^{m+1}+\sigma_{i,j+1}^{m+1}\sigma_{i+1,j}^{m+1})+\\
  		\sum_{i=0}^{2}\sum_{j=0}^{1}(\sigma_{i,j}^m\sigma_{i,j+1}^{m+1}+\sigma_{i,j+1}^m\sigma_{i,j}^{m+1}+\sigma_{j,i}^m\sigma_{j+1,i}^{m+1}+\sigma_{j+1,i}^m\sigma_{j,i}^{m+1}))\\
  		+J_3(\frac{1}{2}\sum_{i=0}^{1}\sum_{j=0}^{1}(\sigma_{i,j}^m\sigma_{i,j+1}^m\sigma_{i+1,j}^m\sigma_{i+1,j+1}^m+\sigma_{i,j}^{m+1}\sigma_{i,j+1}^{m+1}\sigma_{i+1,j}^{m+1}\sigma_{i+1,j+1}^m)+\\
  		\sum_{i=0}^{2}\sum_{j=0}^{1}(\sigma_{i,j}^m\sigma_{i,j+1}^m\sigma_{i,j}^{m+1}\sigma_{i,j+1}^{m+1}+\sigma_{j,i}^m\sigma_{j+1,i}^m\sigma_{j,i}^{m+1}\sigma_{j+1,i}^{m+1})).				
  \end{multline}
   Results are obtained (Fig.7 in the \ref{ex3} example), similar to the results in the \ref{ex1} example and the example \ref{ex2}. 
   
  In the cyclically closed case (Fig. 8 in the \ref{ex3} example) with the Hamiltonian (\ref{hamiltonian_part_gonihedric_3x3xinf_close}):
   
   \begin{multline}\label{hamiltonian_part_gonihedric_3x3xinf_close} 
   		\mathcal H(\sigma)= -\sum_{m=0}^{M-1}(
   		J_{1}(\frac{1}{2}\sum_{i=0}^{2}\sum_{j=0}^{2}(\sigma_{i,j}^m\sigma_{i,j+1}^m+\sigma_{j,i}^m\sigma_{j+1,i}^m+\sigma_{i,j}^{m+1}\sigma_{i,j+1}^{m+1}+\\\sigma_{j,i}^{m+1}\sigma_{j+1,i}^{m+1}+\sigma_{i,j}^m\sigma_{i,j}^{m+1}))
   		+J_2(\sum_{i=0}^{2}\sum_{j=0}^{2}((\sigma_{i,j}^m\sigma_{i+1,j+1}^m+\\\sigma_{i,j+1}^m\sigma_{i+1,j}^m+\sigma_{i,j}^{m+1}\sigma_{i+1,j+1}^{m+1}+\sigma_{i,j+1}^{m+1}\sigma_{i+1,j}^{m+1})/2+\sigma_{i,j}^m\sigma_{i,j+1}^{m+1}+\sigma_{i,j+1}^m\sigma_{i,j}^{m+1}+\\
   		\sigma_{j,i}^m\sigma_{j+1,i}^{m+1}+\sigma_{j+1,i}^m\sigma_{j,i}^{m+1}))
   		+J_3(\sum_{i=0}^{2}\sum_{j=0}^{2}((\sigma_{i,j}^m\sigma_{i,j+1}^m\sigma_{i+1,j}^m\sigma_{i+1,j+1}^m+\\\sigma_{i,j}^{m+1}\sigma_{i,j+1}^{m+1}\sigma_{i+1,j}^{m+1}\sigma_{i+1,j+1}^m)/2+\sigma_{i,j}^m\sigma_{i,j+1}^m\sigma_{i,j}^{m+1}\sigma_{i,j+1}^{m+1}+\\\sigma_{j,i}^m\sigma_{j+1,i}^m\sigma_{j,i}^{m+1}\sigma_{j+1,i}^{m+1})),	
   \end{multline}
   where $\sigma^m_{i,3}=\sigma^m_{i,0},\sigma^m_{3,i}=\sigma^m_{0,i},\sigma^m_{3,3}=\sigma^m_{0,0}, i=0,1,2,$
   obtained results close to the results of the work \cite{Pelizzola}. Therefore, further in \ref{gon_cl} we will consider an analog of the cyclically closed gonigendrial model $2\times2\times\infty$.

\begin{primer} \label{ex3} Let us find the heat capacity in the thermodynamic limit of the gonigendrial model $3\times3\times\infty$ with free boundary conditionsи and for cyclically closed gonigendrial model with $J_{1}=1$, $J_{2}=r$, $J_{3}=0$, $T\in[0.1,5]$, $r\in[-2,2]$ in the absence of an external magnetic field.

\begin{center}
\begin{figure}[h]
	\includegraphics[scale=1]{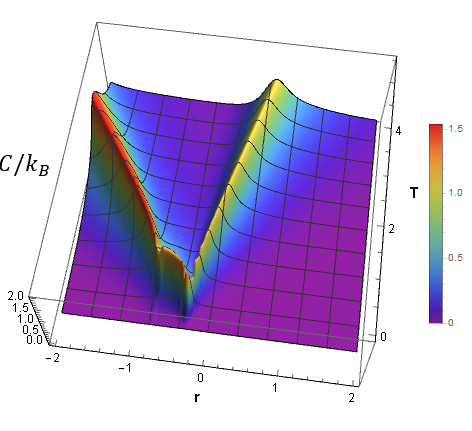}\\
	FIGURE 7. The heat capacity of the $3\times3\times\infty$ gonigendrial model with free boundary conditions for the example \ref{ex3}.
\end{figure}
\end{center}
\begin{center}
\begin{figure}[h]
	\includegraphics[scale=1]{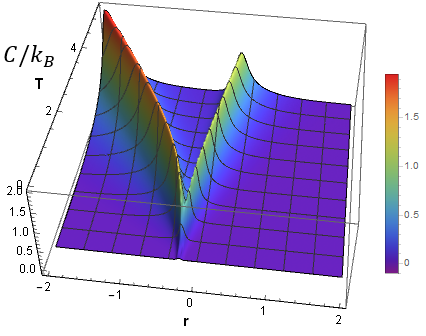}\\
	FIGURE 8. The heat capacity of a cyclically closed gonigendrial model $3\times3\times\infty$ for the example \ref{ex3}.
\end{figure}
\end{center}

For the open model (Fig.7), we observe a shift to the left relative to the value of $r=-0.25$ and a smoothing of the peak in the region of $r=-1$. The closed model, in turn, has a similar appearance to the phase diagram from \cite{Pelizzola}. Therefore, the authors decided to consider an analog of the cyclically closed gonigendrial model for the case $2\times2\times\infty$. 

\end{primer}

\subsection{Exact solution of a cyclically closed gonigendrial model $2\times2\times\infty$}\label{gon_cl}
The general form of the Hamiltonian (\ref{hamiltonian}) allows authors to obtain the exact solution for the analog of the cyclically closed gonigendrial model $2\times2\times\infty$, if we put:

\begin{displaymath}
	\begin{gathered}
		J_{\{t_0^m,t_1^m\}}= J_{\{t_0^{m+1},t_1^{m+1}\}}= J_{\{t_0^{m},t_0^{m+1}\}}=J_1, \\ 8/3 J_{\{t_0^m,t_2^m\}}= 8/3 J_{\{t_0^{m+1},t_2^{m+1}\}}= J_{\{t_0^{m},t_1^{m+1}\}}=  J_{\{t_0^{m},t_3^{m+1}\}}=J_2,\\ 16/3 J_{\{t_0^m,t_1^m,t_2^m,t_3^m\}}= 16/3 J_{\{t_0^{m+1},t_1^{m+1},t_2^{m+1},t_3^{m+1}\}}=  J_{\{t_0^{m},t_1^{m},t_0^{m+1},t_1^{m+1}\}}=J_3,
	\end{gathered}
\end{displaymath}
other coefficients are zero.

The classical notation of the Hamiltonian takes the following form
\begin{equation}\label{hamiltonian_part_gonihedric_cl}
	\begin{gathered} 
		\mathcal H(\sigma)= -\sum_{m=0}^{M-1}(
		J_{1}((\sigma_0^m\sigma_1^m+\sigma_1^m\sigma_2^m+\sigma_2^m\sigma_3^m+\sigma_3^m\sigma_0^m+\sigma_0^{m+1}\sigma_1^{m+1}+\\
		\sigma_1^{m+1}\sigma_2^{m+1}+\sigma_2^{m+1}\sigma_3^{m+1}+\sigma_3^{m+1}\sigma_0^{m+1})+\sigma_0^m\sigma_0^{m+1}+\sigma_1^m\sigma_1^{m+1}+\\
		\sigma_2^m\sigma_2^{m+1}+\sigma_3^m\sigma_3^{m+1}) + J_{2}(3(\sigma_0^m\sigma_2^m+\sigma_1^m\sigma_3^m+\sigma_0^{m+1}\sigma_2^{m+1}+\\
		\sigma_1^{m+1}\sigma_3^{m+1})/4+\sigma_0^m\sigma_1^{m+1}+\sigma_1^m\sigma_2^{m+1}+
		\sigma_2^m\sigma_3^{m+1}+\sigma_3^m\sigma_0^{m+1}+\sigma_0^m\sigma_3^{m+1}+\\
		\sigma_1^m\sigma_0^{m+1}+
		\sigma_2^m\sigma_1^{m+1}+\sigma_3^m\sigma_2^{m+1}) +J_{3}(3(\sigma_0^m\sigma_1^{m}\sigma_2^{m}\sigma_3^{m}+\\
		\sigma_0^{m+1}\sigma_1^{m+1}\sigma_2^{m+1}\sigma_3^{m+1})/8+\sigma_0^m\sigma_1^{m}\sigma_0^{m+1}\sigma_1^{m+1}+\sigma_1^m\sigma_2^{m}\sigma_1^{m+1}\sigma_2^{m+1}+\\
		\sigma_2^m\sigma_3^{m}\sigma_2^{m+1}\sigma_3^{m+1}+\sigma_3^m\sigma_0^{m}\sigma_3^{m+1}\sigma_0^{m+1})).			 
	\end{gathered} 	
\end{equation}

Let us explain this choice of ratios for the coefficients. With cyclic closure, the interaction of the nearest neighbors $J_{\{t_0^m,t_1^m\}}, J_{\{t_0^{m+1},t_1^{m+1}\}}$ doubles, therefore, with taking into account the symmetric notation of the Hamiltonian, the interaction is equated to $J_{\{t_0^{m},t_0^{m+1}\}}$. For the next nearest neighbors, we have an intermediate relation $4/3 J_{\{t_0^m,t_2^m\}}= 4/3 J_{\{t_0^{m+1},t_2^{m+1}\}}=1/2 J_{\{t_0^{m},t_1^{m+1}\}}=1/2  J_{\{t_0^{m},t_3^{m+1}\}}$, where the factor 4 for $J_{\{t_0^m,t_2^m\}},J_{\{t_0^{m+1},t_2^{m+1}\}}$ follows from the symmetric notation and taking into account duplication of interactions in the Hamiltonian (\ref{H_m}), and the divisor is 3 from the increase in the number of interactions during cyclic closure; 1/2 for $J_{\{t_0^{m},t_3^{m+1}\}}$ follows from doubling under cyclic closure. Similarly, one can obtain the relation to plaquettes if we take into account that the Hamiltonian (\ref{H_m}) takes into account the interactions $J_{\{t_0^m,t_1^m,t_2^m,t_3^m\}}, J_{\{t_0 ^{m+1},t_1^{m+1},t_2^{m+1},t_3^{m+1}\}}$ four times.

 Let us introduce the matrices $\tau$ (\ref{tau}) and $\tau'$ (\ref{tau'}) 
\begin{equation}\label{t_gon_cl}
	\begin{gathered}		
		\tau =\begin {psmallmatrix} 
		\begin{tabular}{p{1.5cm} p{3cm} p{2.5cm}  p{1.5cm}} \tiny{
				\begin{center}$p_1^4 q_1^{11/2} p_2^{-5} + p_1^{12} p_2^{11} q_1^{11/2}$ \end{center}}& \tiny{\begin{center}$4 p_1^2 p_2^{-5/2} + 4 p_1^6 p_2^{11/2}$ \end{center}}& \tiny{\begin{center}$4 p_1^4 q_1^{3/2}$ \end{center}}&\tiny{\begin{center} $2 p_2^3 q_1^{-5/2}$\end{center}}\\
			\tiny{\begin{center}$p_1^2 p_2^{-5/2} + p_1^6 p_2^{11/2}$ \end{center}}&\tiny{ \begin{center}$2 q_1^{-11/2} + 2 p_2^{-4} q_1^{-3/2} + 2 p_2^4 q_1^{-3/2} + q_1^{5/2} p_1^{-4} + p_1^4 q_1^{5/2}$ \end{center}}&\tiny{ \begin{center}$2 p_1^{-2} p_2^{-3/2} + 2 p_1^2 p_2^{-3/2}$ \end{center}}&\tiny{\begin{center} $p_1^{-2}p_2^{-5/2} + p_2^{11/2} p_1^{-6}$\end{center}}\\
			\tiny{\begin{center}$2 p_1^4 q_1^{3/2}$ \end{center}}&\tiny{ \begin{center}$4 p_1^{-2} p_2^{-3/2} + 4 p_1^2 p_2^{-3/2}$ \end{center}}&\tiny{ \begin{center}$2 p_2^{-3} q_1^{-5/2} + q_1^{11/2} p_1^{-4} p_2^{-3} + p_1^4 q_1^{11/2} p_2^{-3}$ \end{center}}&\tiny{\begin{center} $2 q_1^{3/2} p_1^{-4}$\end{center}}\\
			\tiny{\begin{center}$2 p_2^3 q_1^{-5/2}$ \end{center}}& \tiny{\begin{center}$4 p_1^{-2} p_2^{-5/2} + 4 p_2^{11/2} p_1^{-6}$ \end{center}}& \tiny{\begin{center}$4 q_1^{3/2} p_1^{-4}$ \end{center}}&\tiny{\begin{center} $q_1^{11/2} p_1^{-4} p_2^{-5} + p_2^{11} q_1^{11/2} p_1^{-12}$\end{center}}
		\end{tabular}
	\end{psmallmatrix},\\	
\end{gathered}
\end{equation}

\begin{equation}\label{t'_gon_cl}
\begin{gathered}		
	\tau' =\begin {psmallmatrix} 
	\begin{tabular}{p{1.5cm} p{3cm} p{2.5cm}} \tiny{
			\begin{center}$p_1^{12} p_2^{11} q_1^{11/2}$ \end{center}}& \tiny{\begin{center}$4 p_1^6 p_2^{11/2}$ \end{center}}& \tiny{\begin{center}$2 p_2^3 q_1^{-5/2}$ \end{center}}\\
		\tiny{\begin{center}$p_1^6 p_2^{11/2}$ \end{center}}&\tiny{ \begin{center}$2p_1^{-2}p_2^{-5/2} + 2p_2^{11/2} p_1^{-6}$ \end{center}}&\tiny{ \begin{center}$p_1^{-2}p_2^{-5/2} + p_2^{11/2} p_1^{-6}$ \end{center}}\\
		\tiny{\begin{center}$p_2^3 q_1^{-5/2}$ \end{center}}&\tiny{ \begin{center}$2 p_1^{-2} p_2^{-5/2} + 2 p_2^{11/2} p_1^{-6}$ \end{center}}&\tiny{ \begin{center}$q_1^{11/2} p_1^{-4} p_2^{-5} + p_2^{11} q_1^{11/2} p_1^{-12}$\end{center}}\\
	\end{tabular}
\end{psmallmatrix},\\	
\end{gathered}
\end{equation}
Now we use Theorem \ref{T1} and Theorem \ref{T2} and write the results for the cyclically closed gonigendrial model
\begin{theorem}\label{S2}

In the thermodynamic limit, the free energy of the gonigendrial model with the Hamiltonian (\ref{hamiltonian_part_gonihedric_cl}) is represented as (\ref{free_energy_lambda}), the internal energy as (\ref{int_energy_lambda}), the heat capacity as (\ref{heat_capacity_lambda}), where $\lambda_{\max}$ is the largest root of the characteristic equation of the matrix $\tau$ (\ref{t_gon_cl}), and for the non-percolation probability, the relation is true (\ref{Percolation1}), where $\mu_{\max}$ is the largest root of the characteristic equation of the matrix $\tau'$ (\ref{t'_gon_cl}).

\end{theorem}

	The coefficients of the characteristic polynomials for $\tau$ (\ref{t_gon}) and $\tau'$ (\ref{t'_gon}) are written out in the formulas (\ref{a})-(\ref{d}) and (\ref{a'})-(\ref{c'}). $\lambda_{max}$ is found using the formulas (\ref{poly4}-\ref{poly4_sol4}), and $\mu_{max}$ is found using the formula (\ref{poly3}).

	\begin{primer}\label{ex4} Let us find the heat capacity of the closed gonigendrial model for $J_{1}=1$, $J_{2}=r$, $J_{3}=0$, $T\in[0.1,5]$, $r\in[- 2,2]$ in the absence of an external magnetic field. One can see the nucleation of phase transition lines 
		, starting at $T>0$, $r=-0.25$ on the heat capacity plot (Fig. 9), as was observed for $3\times3\times\infty$ in the \ref{ex3} example.
	
	\begin{center}
		\begin{figure}[h]
			\includegraphics[scale=1]{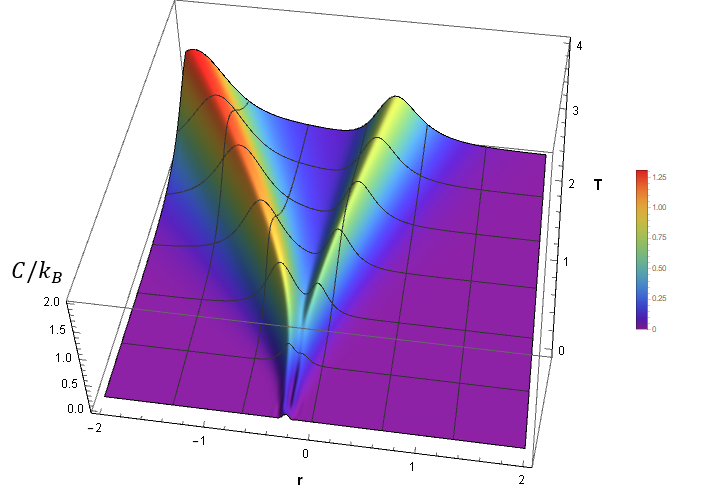}
			
			FIGURE 9. The heat capacity of a closed gonigendrial model for an example \ref{ex4}.
		\end{figure}
	\end{center}

	\begin{center}
		\begin{figure}[h]
			\includegraphics[scale=1]{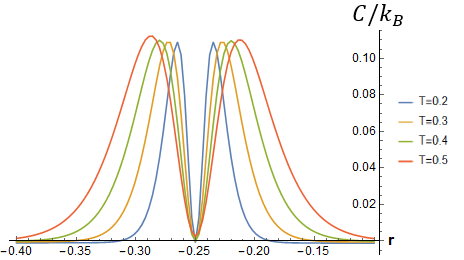}\\
			FIGURE 10. The heat capacity of a closed gonigendrial model for an example \ref{ex4} for $T=0.2$, $T=0.3$, $T=0.4$, $T=0.5$.
		\end{figure}
	\end{center}

Here one can already see a great similarity with the phase pattern for the gonigendrial model with similar interactions in the work \cite{Pelizzola}.

\end{primer}

	\section{Antiferromagnetic layered Ising model }\label{NN_NNN_interactions}
	
	We consider the Hamiltonian that takes into account only the interactions of nearest neighbors and interactions of next nearest neighbors, and $ 2 J_{\{t_0^m,t_1^m\}}= 2 J_{\{t_0^{m+1},t_1^{m+1}\}}=J_{\{t_0^{m},t_0^{m+1}\}}=J_{1}$, $ 4 J_{\{t_0^m,t_2^m\}}= 4 J_{\{t_0^{m+1},t_2^{m+1}\}}=J_{\{t_0^{m},t_1^{m+1}\}}=J_{\{t_0^{m},t_3^{m+1}\}}=J_{2}$,and the rest of the coefficients are zero. Numerical calculations based on the replica algorithm by the Monte Carlo method for a lattice model with such a Hamiltonian are made in the work \cite{Ramazanov}. We write the Hamiltonian (\ref{hamiltonian}) for this case as 
	
	\begin{equation}\label{hamiltonian_anti}
		\begin{gathered} 
				\mathcal H(\sigma)= -\sum_{m=0}^{M-1}(
			J_{1}((\sigma_0^m\sigma_1^m+\sigma_1^m\sigma_2^m+\sigma_2^m\sigma_3^m+\sigma_3^m\sigma_0^m+\sigma_0^{m+1}\sigma_1^{m+1}+\\
			\sigma_1^{m+1}\sigma_2^{m+1}+\sigma_2^{m+1}\sigma_3^{m+1}+\sigma_3^{m+1}\sigma_0^{m+1})/2+\sigma_0^m\sigma_0^{m+1}+\sigma_1^m\sigma_1^{m+1}+\\
			\sigma_2^m\sigma_2^{m+1}+\sigma_3^m\sigma_3^{m+1}) + J_{2}((\sigma_0^m\sigma_2^m+\sigma_1^m\sigma_3^m+\sigma_0^{m+1}\sigma_2^{m+1}+\\
			\sigma_1^{m+1}\sigma_3^{m+1})/2+\sigma_0^m\sigma_1^{m+1}+\sigma_1^m\sigma_2^{m+1}+
			\sigma_2^m\sigma_3^{m+1}+\sigma_3^m\sigma_0^{m+1}+\sigma_0^m\sigma_3^{m+1}+\\
			\sigma_1^m\sigma_0^{m+1}+
			\sigma_2^m\sigma_1^{m+1}+\sigma_3^m\sigma_2^{m+1})).					    			
		\end{gathered} 	
	\end{equation}

	It can be seen that the considering case can be obtained from the Hamiltonian (\ref{hamiltonian_part_gonihedric}), by setting $J_3=0$.
	
	Let us use Theorem \ref{T1} and Theorem \ref{T2}, and write 
	results for antiferromagnetic layered Ising model.

\begin{theorem}\label{S3}

	In the thermodynamic limit, the free energy of the model with the Hamiltonian (\ref{hamiltonian_anti}) 
  can be represent in the form (\ref{free_energy_lambda}), the internal energy can be represented in the form (\ref{int_energy_lambda}), the heat capacity can be represented in the form (\ref{heat_capacity_lambda}), wherein  $\lambda_{\max}$ is the largest root of the equation (\ref{poly4}),
	where 
	\begin{multline}{\label{a_part1}}
	a=-2 - p_1^{-4} - p_1^4 - 2 p_2^{-6} - 2 p_2^{-4} - 2 p_2^{-2} - 
	p_1^{-4} p_2^{-2} - p_1^4 p_2^{-2} - 2 p_2^4 - p_2^{10} p2^{-8} -\\ p2^8 p4^{10},
\end{multline}
\begin{multline}{\label{b_part1}}
	b=-8 p_1^{-4 }-
	8 p_1^{4 }+
	1 p_2^{-12 }+ 
	4 p_2^{-10 }+ 
	4 p_2^{-8 }+
	2 p_1^{-4 }p_2^{-8 }+ 
	2 p_1^{4}p_2^{-8 }+
	4 p_1^{-4 }p_2^{-6 }+\\
	4 p_1^{4}p_2^{-6 }-
	6 p_2^{-2 }+
	1 p_1^{-8 }p_2^{-2 }-
	4 p_1^{-4 }p_2^{-2 }- 
	4  p_1^{4}p_2^{-2 }+
	p_1^{8}p_2^{-2 }+
	4  p_2^{2 }-
	6  p_2^{2}p_1^{-4 }- 
	6  p_1^{4 }p_2^{2 }- \\
	4  p_2^{4 }+
	p_2^{4}p_1^{-8 }+
	p_1^{8 }p_2^{4 }+
	2 p_2^{6}p_1^{-8 }+
	2 p_1^{8 }p_2^{6 }+
	p_2^{8}p_1^{-12 }+ 
	2 p_2^{8}p_1^{-8 }+
	p_2^{8}p_1^{-4 }+ 
	p_1^{4 }p_2^{8 }+
	2 p_1^{8 }p_2^{8 }+\\ 
	p_1^{12 }p_2^{8 }+ 
	p_2^{10}p_1^{-12 }-
	2 p_2^{10}p_1^{-8 }+
	p_2^{10}p_1^{-4 }+
	p_1^{4 }p_2^{10 }-
	2 p_1^{8 }p_2^{10 }+
	p_1^{12 }p_2^{10 }+
	2 p_2^{14}p_1^{-8 }+
	2 p_1^{8 }p_2^{14 }+\\
	p_2^{20},
\end{multline}
\begin{multline}{\label{c_part1}}
	c=40 +
	14p_1^{-8 }+
	32p_1^{-4 }+
	32 p_1^{4 }+ 
	14 p_1^{8 }-
	2p_2^{-16 }-
	2p_2^{-14 }-
	1p_1^{-4 }p_2^{-14 }-
	p_1^{4}p_2^{-14 }-\\
	2p_2^{-12 }-
	5p_1^{-4 }p_2^{-12 }- 
	5 p_1^{4}p_2^{-12 }+
	34p_2^{-8 }-
	2p_1^{-8 }p_2^{-8 }+
	16p_1^{-4 }p_2^{-8 }+
	16 p_1^{4}p_2^{-8 }-\\
	2 p_1^{8}p_2^{-8 }+
	8p_1^{-4 }p_2^{-6 }+
	8 p_1^{4}p_2^{-6 }-
	88p_2^{-4 }-
	12p_1^{-4 }p_2^{-4 }-
	12 p_1^{4}p_2^{-4 }-
	24 p_2^{2 }-
	p_2^{2}p_1^{-12 }-\\
	2 p_2^{2}p_1^{-8 }+
	3 p_2^{2}p_1^{-4 }+ 
	3 p_1^{4 }p_2^{2 }-
	2 p_1^{8 }p_2^{2 }-
	p_1^{12 }p_2^{2 }+
	8 p_2^{4 }-
	3 p_2^{4}p_1^{-12 }-
	30 p_2^{4}p_1^{-8 }-
	31 p_2^{4}p_1^{-4 }-\\
	31 p_1^{4 }p_2^{4 }-
	30 p_1^{8 }p_2^{4 }-
	3 p_1^{12 }p_2^{4 }+
	6 p_2^{8 }-
	p_2^{8}p_1^{-16 }+
	8 p_2^{8}p_1^{-12 }+
	16 p_2^{8}p_1^{-8 }+ 
	8 p_2^{8}p_1^{-4 }+
	8 p_1^{4 }p_2^{8 }+\\
	16 p_1^{8 }p_2^{8 }+
	8 p_1^{12 }p_2^{8 }-
	p_1^{16 }p_2^{8 }+
	8 p_2^{10}p_1^{-4 }+
	8 p_1^{4 }p_2^{10 }-
	16 p_2^{12 }-
	2 p_2^{12}p_1^{-12 }-
	4 p_2^{12}p_1^{-8 }+
	6 p_2^{12}p_1^{-4 }+\\
	6 p_1^{4 }p_2^{12 }-
	4 p_1^{8 }p_2^{12 }-
	2 p_1^{12 }p_2^{12 }-
	2 p_2^{16 }-
	2 p_2^{18 }-
	p_2^{18}p_1^{-4 }-
	p_1^{4 }p_2^{18 }+
	6 p_2^{20 }-
	p_2^{20}p_1^{-4 }- 
	p_1^{4 }p_2^{20 }-\\
	2 p_2^{24},
\end{multline}
\begin{multline}{\label{d_part1}}
	d=4p_2^{-18 }+
	2p_1^{-4 }p_2^{-18 }+ 
	2 p_1^{4}p_2^{-18 }- 
	26p_2^{-14 }+
	p_1^{-8 }p_2^{-14 }-
	12p_1^{-4 }p_2^{-14 }-
	12 p_1^{4}p_2^{-14 }+ \\
	p_1^{8}p_2^{-14 }+
	68p_2^{-10 }+
	18p_1^{-4 }p_2^{-10 }+
	18 p_1^{4}p_2^{-10 }- 
	64p_2^{-6 }-
	16p_1^{-8 }p_2^{-6 }+
	16p_1^{-4 }p_2^{-6 }+ 
	16 p_1^{4}p_2^{-6 }- \\
	16 p_1^{8}p_2^{-6 }- 
	80p_2^{-2 }+
	2 p_1^{-12 }p_2^{-2 }+
	36p_1^{-8 }p_2^{-2 }-
	54p_1^{-4 }p_2^{-2 }- 
	54 p_1^{4}p_2^{-2 }+ 
	36 p_1^{8}p_2^{-2 }+ \\
	2 p_1^{12}p_2^{-2 }+ 
	266 p_2^{2 }+
	p_2^{2}p_1^{-16 }-
	12 p_2^{2}p_1^{-12 }- 
	14 p_2^{2}p_1^{-8 }+ 
	4 p_2^{2}p_1^{-4 }+ 
	4 p_1^{4 }p_2^{2 }- 
	14 p_1^{8 }p_2^{2 }-\\ 
	12 p_1^{12 }p_2^{2 }+
	p_1^{16 }p_2^{2 }-
	80 p_2^{6 }+
	2 p_2^{6}p_1^{-12 }+
	36 p_2^{6}p_1^{-8 }- 
	54 p_2^{6}p_1^{-4 }-
	54 p_1^{4 }p_2^{6 }+
	36 p_1^{8 }p_2^{6 }+
	2 p_1^{12 }p_2^{6 }- \\
	64 p_2^{10 }-
	16 p_2^{10}p_1^{-8 }+
	16 p_2^{10}p_1^{-4 }+
	16 p_1^{4 }p_2^{10 }- 
	16 p_1^{8 }p_2^{10 }+ 
	68 p_2^{14 }+
	18 p_2^{14}p_1^{-4 }+
	18 p_1^{4 }p_2^{14 }-\\
	26 p_2^{18 }+
	p_2^{18}p_1^{-8 }-
	12 p_2^{18}p_1^{-4 }-
	12 p_1^{4 }p_2^{18 }+
	p_1^{8 }p_2^{18 }+
	4 p_2^{22 }+
	2 p_2^{22}p_1^{-4 }+ 
	2 p_1^{4 }p_2^{22},	
\end{multline}
and for the non-percolation probability, the relation is true (\ref{Percolation1}), where $\mu_{\max}$ is the largest root of the equation (\ref{t'_gon}), wherein
\begin{equation}{\label{a'_part1}}
	a'=-1 - p_1^4 - p_2^{-6} - 2 p_2^4 - p_2^{10} p_1^{-8} - p_1^8 p_2^{10},
\end{equation}
\begin{multline}{\label{b'_part1}}
	b'=-p_2^{-6} + p_1^4 p_2^{-6} + 2 p_2^{-2} - 4 p_2^2 p_1^{-4} - 2 p_2^4 + 
	p_1^8 p_2^4 - p_2^{10} p_1^{-8} + p_2^{10} p_1^{-4} - 3 p_1^8 p_2^{10} + p_1^{12} p_2^{10} +\\ 
	2 p_2^{14} p_1^{-8} + 2 p_1^8 p_2^{14} + p_2^{20},
\end{multline}
\begin{multline}\label{c'_part1}
	c'=2 p_2^4 - 6 p_1^4 p_2^4 + 5 p_1^8 p_2^4 - p_1^{12} p_2^4 + 4 p_2^8 - 
	2 p_1^8 p_2^8 - 8 p_2^{12} + 4 p_1^4 p_2^{12} + 5 p_2^{20} -\\ p_1^4 p_2^{20} - 2 p_2^{24},
\end{multline}

\end{theorem}

	The coefficients of the characteristic polynomials for $\tau$ (\ref{t_gon}) and $\tau'$ (\ref{t'_gon}) are written in the formulas (\ref{a})-(\ref{d}) and (\ref{a'})-(\ref{c'}). $\lambda_{max}$ is found using the formulas (\ref{poly4}-\ref{poly4_sol4}), аnd $\mu_{max}$ is found using the formula (\ref{poly3}).
	
	\begin{primer}\label{ex5} Let us find the heat capacity for $J_{1}=-1$, $J_{2}=-r$, $T\in[0.2,0.5]$, $r\in[0,1]$ in the absence of an external magnetic field. Fig. 11 shows the nucleation of lines of a phase transition of the 1st order at $r=0.65$.

	\begin{figure}[h]
		\includegraphics[scale=1]{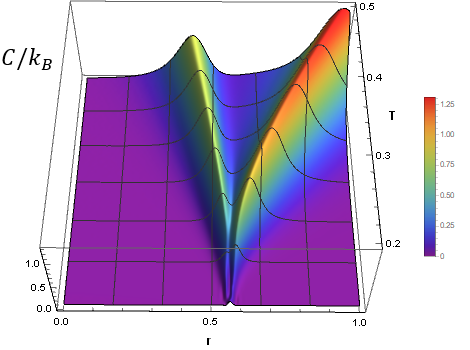}
		\\FIGURE 11. The heat capacity graph for the example \ref{ex5}
		\label{fig:C_part2}
	\end{figure}
	
		\end{primer}
	
	In the article \cite{Ramazanov} for the considered case, a phase diagram is constructed, which has a similar form, but the phase transition line comes from $r=0.5$, and not from $r=0.65$, as obtained in this paper. The shift can be explained by the influence of boundary conditions (the authors consider free boundary conditions), which have a strong effect on the model in the considered volume $ 2\times2\times\infty$.
	
	Note that the exact solution for the anisotropic Hamiltonian, which takes into account only the interaction of nearest neighbors, was obtained in \cite{Yurishchev}. This solution coincides with the results of this paper if we assume isotropy $J_x=J_y=J_z=J_{1}$ in the \cite{Yurishchev} model  and put $J_{2}=0$ in the Hamiltonian (\ref{hamiltonian_anti}). Then the coefficients of the 4th degree polynomial (\ref{poly4}) have the form:
	
	\begin{equation}{\label{a_yurishchev}}
		a=-10 -4\cosh(4 J_{1}) - 2\cosh(8 J_{1}),
	\end{equation}
	\begin{equation}{\label{b_yurishchev}}
		b=4 - 20\cosh(4 J_{1}) + 12\cosh(8 J_{1}) + 4\cosh(12 J_{1}),
	\end{equation}
	\begin{multline}\label{c_yurishchev}
		c=-46 + 60\cosh(4 J_{1}) - 16\cosh(8 J_{1}) + 4\cosh(12 J_{1}) -\\ 2\cosh(16 J_{1}),
	\end{multline}	
	\begin{multline}\label{d_yurishchev}
		d=70 - 112\cosh(4 J_{1}) + 56\cosh(8 J_{1}) - 16\cosh(12 J_{1}) +\\ 2\cosh(16 J_{1}).
	\end{multline}
	
	The coefficients (\ref{a_yurishchev})-(\ref{d_yurishchev}) coincide with the coefficients of the polynomial in \cite{Yurishchev}, and hence the largest eigenvalue also coincides. In \cite{Yurishchev} the solution is written out taking into account the special form of the resulting polynomial, which allows one to decompose a fourth degree polynomial into two second degree polynomials.

	\section{Conclusion}\label{conclusion}
	Thus, the authors obtained the exact value in the thermodynamic limit of free energy, heat capacity, and non-percolation in the $2\times2\times\infty$ strip for a family of Hamiltonians that are invariant under rotation about the central symmetry axis by $\pi/2$. The problem is reduced to finding the root of the fourth degree polynomial, which is the largest eigenvalue of the transfer matrix. The value of the non - percolation probability of the model in the case of percolation, which has the same invariance, is found. Some types of Hamiltonians are considered for which the largest eigenvalue of the transfer matrix is the root of a quadratic trinomial or a linear equation. The authors note that other Hamiltonians can be found in a similar way, for which a similar expansion is possible. An exact value is obtained for the free energy and heat capacity for the $2\times2\times\infty$ gonigendrial model in the case of free boundary conditions and in the case of an analog of cyclic closure. The similarity of the numerical solution in the $3\times3\times\infty$ strip and the results of the work of other authors with the exact solution for the analog of cyclic closure in the $2\times2\times\infty$ strip indicates that the authors managed to choose the analog of cyclic closure correctly. The Hamiltonian is written out taking into account a special selection of interaction coefficients simulating cyclic closure. The exact value of the free energy and heat capacity in the thermodynamic limit in the $2\times2\times\infty$ strip is obtained for the Antiferromagnetic layered Ising model, the results are compared with the exact solution in the work \cite{Yurishchev_1997}.

\end{document}